\documentclass[12pt]{article}
\usepackage[T1]{fontenc}
\usepackage[utf8]{inputenc}
\usepackage{enumerate}
\usepackage{url} 
\usepackage[top=0.95in, bottom=0.95in, left=1in, right=1in, letterpaper]{geometry}
\usepackage[leqno]{amsmath}
\usepackage{amsthm}
\usepackage{graphicx}
\usepackage{setspace}
\usepackage{natbib}
\usepackage{amssymb}
\usepackage[hidelinks]{hyperref}
\usepackage{multirow}
\usepackage{amsfonts}
\usepackage{mathtools}
\usepackage{todonotes}
\usepackage{xr}
\usepackage{comment}

\usepackage{bm}
\usepackage{framed}
\usepackage{pbox}
\usepackage{color}

\usepackage[linesnumbered,ruled,vlined]{algorithm2e}
\SetKwInput{KwInput}{Input}                
\SetKwInput{KwOutput}{Output}              

\newtheorem{theorem}{Theorem}
\newtheorem{assumption}{Assumption}

\newtheorem{lemma}{Lemma}

\theoremstyle{remark}

\newtheorem*{example*}{Example}

\numberwithin{theorem}{section}
\numberwithin{lemma}{section}
\numberwithin{equation}{section}
\newcommand{\norm}[1]{\left\Vert#1\right\Vert}
\newcommand{\abs}[1]{\left\vert#1\right\vert}

\newcommand{\blind}{0}

\defcitealias{CL34}{Clopper--Pearson (1934)}	
\defcitealias{Raic2019}{Rai\v{c}'s (2019)}

\begin{document}

\def\spacingset#1{\renewcommand{\baselinestretch}%
{#1}\small\normalsize} \spacingset{1}


\if0\blind
{
  \title{\bf Inference in a class of optimization problems: Confidence regions and finite sample bounds on errors in coverage probabilities}
  \author{Joel L. Horowitz\thanks{
An earlier version of this paper was presented at the conference ``Incomplete Models,'' which took place at Northwestern University in November 2018. We thank the conference participants, the associate editor, and two anonymous referees for helpful comments. This research was supported in part by the European Research Council (Grant No. ERC-2014-CoG-646917-ROMIA) and the UK Economic and Social Research Council (Grant No. ES/P008909/1 to the CeMMAP). Part of Joel L. Horowitz's research was carried out while he was a visitor to the Department of Economics at University College London.
}\hspace{.2cm}\\
    Department of Economics, Northwestern University\\
    and \\
    Sokbae Lee \\
    Department of Economics, Columbia University}
  \maketitle
} \fi

\if1\blind
{
  \bigskip
  \bigskip
  \bigskip
  \begin{center}
    {\LARGE\bf Inference in a class of optimization problems: Confidence regions and finite sample bounds on errors in coverage probabilities}
\end{center}
  \medskip
} \fi

\bigskip
\begin{abstract}
This paper describes three methods for carrying out non-asymptotic inference on partially identified parameters that are solutions to a class of optimization problems. Applications in which the optimization problems arise include estimation under shape restrictions, estimation of models of discrete games, and estimation based on grouped data. The partially identified parameters are characterized by restrictions that involve the unknown population means of observed random variables in addition to structural parameters. Inference consists of finding confidence intervals for functions of the structural parameters. Our theory provides finite-sample lower bounds on the coverage probabilities of the confidence intervals under three sets of assumptions of increasing strength. With the moderate sample sizes found in most economics applications, the bounds become tighter as the assumptions strengthen. We discuss estimation of population parameters that the bounds depend on and contrast our methods with alternative methods for obtaining confidence intervals for partially identified parameters. The results of Monte Carlo experiments and empirical examples illustrate the usefulness of our method.
\end{abstract}

\noindent%
\vfill

\newpage
\spacingset{1.45} 

\doublespacing



\section{Introduction}
\label{sec:intro}

This paper presents three methods for carrying out non-asymptotic inference about
a function of partially identified structural parameters of an econometric model.
The methods apply to models that impose shape restrictions \citep[e.g.,][]{FH:2015, HL:17}, a variety of partially identified models \citep[e.g.,][]{Manski:book, Tamer:10} that include discrete games \citep[e.g.,][]{Ciliberto:Tamer:09},
and models in which a continuous function is inferred from the average values of variables in a finite number of discrete groups \citep[e.g.,][]{BDM:98,KT16}. The specific inference problem consists of finding upper and lower bounds on the partially identified function
$f(\psi)$ under the restrictions $g_1(\psi,\mu) \leq 0$ and $g_2(\psi,\mu) = 0$, where $\psi$ is a vector of structural parameters, $\mu$ is a vector of unknown population means of observable random variables, $f$ is a known, real-valued function, and $g_1$ and $g_2$ are known possibly vector-valued functions. The inequality $g_1(\psi,\mu) \leq 0$ holds component-wise.

Most existing methods for inference in our framework are based on asymptotic approximations.  
They provide correct inference in the limit of an infinite sample size but do not provide information about the accuracy of the asymptotic approximations in finite samples. We provide three methods for obtaining finite-sample lower bounds on the coverage probability of a confidence interval for $f (\psi)$.
One method uses asymptotic approximations to obtain a
confidence interval. The other two methods do not use asymptotic approximations. All the methods provide information about the accuracy of finite-sample inference.

There are several approaches to carrying out non-asymptotic inference  in our framework.
Sometimes a statistic with a known finite-sample distribution makes finite-sample inference possible. For example, the \citetalias{CL34} confidence interval for a population probability is obtained by inverting the binomial probability distribution function. \citet{Manski:07} used the Clopper-Pearson interval to construct finite-sample confidence sets for counterfactual choice probabilities. Our methods apply to parameters that are not necessarily probabilities. A second existing method consists of using Hoeffding's inequality to obtain a confidence interval. \citet{STZ:2018} used this inequality to construct a confidence interval for a partially identified population moment. Hoeffding's inequality requires the underlying random variable to have a known bounded support. Our methods do not require the underlying random variable to have a known or bounded support.
\citet{Minsker2015} developed a confidence set for a
vector of population means by using a method called ``median of means.'' The bounds provided by this method are looser than the bounds provided by our methods. In addition, Minsker's method depends on certain user-selected tuning parameters. There are no data-based, efficient ways to choose these parameters in applications.

Our first method consists of making a normal approximation to the unknown distribution of the sample average. This method makes certain assumptions about low-order moments of the underlying random variable but does not restrict its distribution in other ways. A variety of results provide finite-sample upper bounds on the errors made by normal approximations. The  Berry-Ess\'{e}en inequality for the average of a scalar random variable is a well-known example of such a bound. \citet{Bentkus03} provides a bound on the error of a multivariate normal approximation to the distribution of the sample average of a random vector.
Other normal approximations for random vectors are given by \citet{spokoiny2015}; \citet{CCK:17}; and \citet{zhilova2020}; among others.  Our first method uses a bound on the error of the multivariate normal approximation that is due to \citet{Raic2019}. \citetalias{Raic2019} bound is a refined and tighter version of the bound of \citet{Bentkus03}. 

The bound of \citet{CCK:17} may be tighter than that of \citet{Raic2019} when the dimension of $\mu$ exceeds the
sample size, but the bound of \citet{Raic2019} is tighter when the dimension of $\mu$ is small compared
to the sample size, which is the case we treat in this paper. In contrast to conventional asymptotic inference approaches, our first method provides a finite-sample lower bound on the coverage probability of a confidence interval for the partially identified function $f(\psi)$.

The bound provided by our first method is loose in samples of the moderate sizes that occur in most economics applications, though not necessarily in very large samples. This is because it places only weak restrictions on the distribution of the underlying random variable, which may be far from normal. Our second method obtains a tighter bound in moderate size samples by assuming that the distributions of the components of the possibly vector-valued underlying random variable are sub-Gaussian. The sub-Gaussian assumption places stronger restrictions on the thickness of the tails of the relevant distributions than do the assumptions of the bound based on \citetalias{Raic2019} inequality. Our third method tightens the bound obtained with our second method by assuming that if the underlying random variable is vector-valued, then its distribution is sub-Gaussian in a vector sense. This assumption is stronger than the assumption that the components of a random
vector are individually sub-Gaussian. The bounds obtained with the second and third methods are identical if the underlying random variable is a scalar.

The bounds provided by all the methods depend on unknown population parameters. This dependence is unavoidable and can be removed only in special cases. The parameters of the bounds of the second and third methods can be estimated, however, which makes it possible to estimate the bounds in applications. We describe how to do this.
The resulting estimated bounds are asymptotic. They do not have finite-sample validity but can provide useful, though possibly rough, indications of the magnitudes of the finite-sample bounds. We present the results of Monte Carlo experiments that illustrate the relation between the exact finite-sample bounds and the consistent estimates.

Our work is broadly related to the literature on inference in partially identified models. 
\citet{Tamer:10}, \citet{canay_shaikh_2017}, \citet{ho_rosen_2017}, and \citet{Molinari} provide recent surveys. \citet{CCT:18} describe a Monte Carlo method
 for carrying out asymptotic inference for a class of models that includes our framework. 
 \cite{BCS:17}
and 
 \cite{KMS19} develop asymptotic inference methods for  subvectors  of partially identified parameters in moment inequality models.
 \citet{CCK} and \citet{BBC} construct confidence regions by inverting pointwise tests of a hypothesis about the (sub)vector of parameters that are partially identified by a large number of moment inequalities.
The inference problem we treat is different from those in the foregoing papers in that
 we focus
 on inference about parameters that are solutions to a class of optimization problems
  that is different from moment inequality problems.
 Our methods and results do not apply to moment inequalities.
\citet{Kline:Tamer} describe Bayesian inference in a class of models that includes a special case of the models we treat.
 Two more closely related papers are
\citet{HSS:2017} and \citet{Shi:Shum:15}, who propose a method for asymptotic inference about estimators defined by mathematical programs.
However, the class of estimation problems they treat is different from ours and overlaps with ours only under highly restrictive assumptions about both classes. \citet{HSS:2017} 
and \citet{Shi:Shum:15} do not provide finite-sample bounds on the errors of their asymptotic approximations. 

 Our work is also related to the econometrics literature on finite-sample inference. \,
\citet{STZ:2018} consider finite-sample inference in auction models. Their framework and method are very different from those in this paper. 
  In a different context, \cite{CHJ:09} and \cite{Rosen:Ura:19} propose finite-sample inference for quantile regression models and for the maximum score estimand, respectively.  Their methods and the classes of models they treat are distinct from ours.

The remainder of this paper is organized as follows. Section \ref{sec:toy:exmp} 
describes the inferential problem we treat, our methods for obtaining a confidence interval for $f(\psi)$, and the three methods
for obtaining a finite-sample lower bound on the coverage probability of a confidence interval.
Section \ref{sec:toy:exmp} 
also describes two empirical studies that illustrate how the inferential problem arises in applications.
 Section \ref{sec:algorithms} describes computational procedures for implementing our methods. 
 Section \ref{sec:Exmp:new}  presents
 an empirical application of the methods.
  Section \ref{sec:MC:new} reports the results of a Monte Carlo investigation of the numerical performance of our methods, and Section  \ref{sec:concl} gives concluding comments. 
 The proofs of theorems are presented in online Appendix \ref{app:proof}. 
 Online Appendices \ref{subsec:cont}--\ref{sec:Exmp:appendix} provide additional technical information about our methods, a description of Minsker's (2015) method, an additional empirical application, and additional Monte Carlo results.

\section{The Method}\label{sec:toy:exmp}	

Section \ref{subsec:inference} presents an informal description of inferential problem we address.
 Section \ref{subsec:examples} gives two examples of empirical applications in which the inferential problem arises. Section \ref{subsec:analysis} provides a formal description of our methods for constructing confidence intervals and bounds on coverage probabilities. 


\subsection{The Inferential Problem}\label{subsec:inference}

Let
 $\{ X_i : i=1,\ldots,n \}$ be an independent random sample from the distribution of the random vector  $X \in \mathbb{R}^p$
 for some finite $p \geq 1$. 
Define $\mu = \mathbb{E}(X)$ and $\Sigma = \text{cov}(X)$.
We assume that both exist. 
Let $\psi$ be a finite-dimensional parameter and $f (\psi )$ be a real-valued, known function. We assume throughout this section that $f (\psi )$ is only partially identified by the sampling process, though our results also hold if
$f (\psi )$ is point identified. We seek a confidence interval for $f (\psi )$,
which we define as a data-based interval that contains $f (\psi )$ with probability exceeding a known value. Let $g_1(\psi , \mu)$ and $g_2 (\psi , \mu)$ be
possibly vector valued known functions satisfying 
$g_1(\psi, \mu) \leq 0$ and $g_2(\psi, \mu) = 0$ component-wise.
Define
\begin{align}\label{obj-pop}
J_{+} := \max_{\psi} f(\psi)
 \; \text{ and } \; 
J_{-} := \min_{\psi} f(\psi)
\end{align}
subject to the component-wise constraints:
\begin{align}\label{g-constraint}
g_1(\psi, \mu) \leq 0, \;
g_2(\psi, \mu) = 0, \;
\psi &\in \Psi,
\end{align}
where $\Psi$ is a compact parameter set.
Online Appendix \ref{subsec:cont} extends \eqref{obj-pop}-\eqref{g-constraint} to the case in which  
$g_1$ and $g_2$ depend on a continuous covariate in addition to $\psi$ and $\mu$.

We are interested in the identification interval $J_{-} \leq f(\psi) \leq J_{+}$.
However, this interval cannot be calculated in applications because $\mu$ is unknown. Therefore, we estimate $\mu$ by the
sample average $\bar{X} = n^{-1} \sum_{i=1}^n X_i$, and we estimate $J_{+}$ and $J_{-}$ by
\begin{align}\label{obj0}
\hat{J}_{+}(\bar{X}) := \max_{\psi, m} f(\psi)
 \; \text{ and } \; 
\hat{J}_{-}(\bar{X}) := \min_{\psi, m} f(\psi)
\end{align}
subject to the constraints
\begin{subequations}\label{g-constraint-est}
\begin{align}
&g_1(\psi, m) \leq 0, 
g_2(\psi, m) = 0, 
\psi \in \Psi, \label{g-constraint-est-1} \\
&n^{1/2} (\bar{X} - m) \in \mathcal{S}, \label{S-constraint-est}
\end{align}
\end{subequations}
where $\mathcal{S}$ is 
a set for which $n^{1/2} (\bar{X} - \mu) \in \mathcal{S}$ with high probability.
Since $\mu$ is unknown, we replace it with the variable of optimization $m$ in \eqref{obj0}--\eqref{g-constraint-est-1} but require $m$ to satisfy \eqref{S-constraint-est}.
The resulting confidence interval for $f(\psi)$ is
\begin{align}\label{est-ci}
\hat{J}_{-}(\bar{X}) \leq f(\psi) \leq \hat{J}_{+}(\bar{X}).
\end{align}
This is also a confidence interval for the identified set containing $f(\psi)$.
Section \ref{subsec:analysis} provides three different finite-sample lower bounds on the probability that this interval contains $f(\psi)$. That is, Section \ref{subsec:analysis} provides three finite-sample lower bounds on
\begin{align}\label{est-ci-prob}
\mathbb{P} \left[ \hat{J}_{-}(\bar{X}) \leq J_{-} \leq f(\psi)  \leq J_{+} \leq \hat{J}_{+}(\bar{X}) \right].
\end{align}
The three bounds correspond to increasingly strong assumptions about the distribution of $X$ and are increasingly tight with samples of the moderate sizes found in most economics applications, though not necessarily with very large samples.

The two leading examples of $\mathcal{S}$ in constraint \eqref{S-constraint-est} are a box and an ellipsoid.
If $\mathcal{S}$ is a box,  let $D$ be a diagonal matrix whose diagonal elements are strictly positive.
For example, $D$ might be the diagonal elements of $\Sigma$ if $\Sigma$ is known
or the diagonal elements of a consistent estimate, $\widehat{\Sigma}$, if $\Sigma$ is unknown. 
Choose $\kappa_b(1-\alpha)$ so that 
the following holds, uniformly in $j = 1,\ldots,p$, with probability $1-\alpha$:
\begin{align*}
\left| \bar{X}_j - \mu_j \right|
\leq  \kappa_b(1-\alpha) ( n D_{j} )^{1/2},
\end{align*}
where the subscript $j$ denotes the $j$'th component of a vector or the $(j, j)$ component of a matrix.
In this case, \eqref{S-constraint-est} becomes $p$ constraints
and can be viewed as a sample analog of 
$\left| \mathbb{E}(X_j) - \mu_j \right| \leq 0$ 
with a relaxed constraint on $\mu_j$
for each $j = 1,\ldots,p$.
Section~\ref{sec:algorithms} presents methods for choosing  $\kappa_b(1-\alpha)$.

If  $\mathcal{S}$ is an ellipsoid, 
let $\Upsilon$ denote a positive definite $p \times p$ matrix, possibly
$\Sigma$ or $\widehat{\Sigma}$ if those matrices are non-singular, or the identity matrix. Choose 
$\kappa_e(1-\alpha)$ so that 
$$
n (\bar{X} - \mu)' \Upsilon^{-1} (\bar{X} - \mu) 
\leq \kappa_e(1-\alpha)
$$
with probability $1-\alpha$.
In this case, \eqref{S-constraint-est} is a single constraint.
Section~\ref{sec:algorithms}  presents a method for choosing  $\kappa_e(1-\alpha)$.
 When $\Sigma$ is difficult to estimate or is singular, we may use a sphere by choosing a critical value $\kappa_s(1-\alpha)$ such that 
\begin{align*}
n (\bar{X} - \mu)'  D^{-1} (\bar{X} - \mu) 
\leq \kappa_s(1-\alpha)
\ \ \text{ or } \ \
n (\bar{X} - \mu)'  (\bar{X} - \mu) 
\leq \kappa_s(1-\alpha)
\end{align*}
with probability $1-\alpha$.
In general, the implementation of our methods  is simpler if $\mathcal{S}$ is indexed by a scalar critical value $\kappa(1-\alpha)$.


It is straightforward to allow the objective function $f(\psi)$  to depend on $\mu$.
For the lower bound $\hat{J}_{-}(\bar{X})$, we introduce an auxiliary variable $t$
that acts as an upper bound on $f(\psi, \mu)$ and solve:
$
 \min_{\psi,m,t} t
$
subject to $f(\psi, m) \leq t$ and \eqref{g-constraint-est}.
For the upper bound 
$\hat{J}_{+}(\bar{X})$, we introduce a lower bound $s$ on  $f(\psi, \mu)$ and solve:
$
\max_{\psi,m,s} s
$
subject to $f(\psi, m) \geq s$ and \eqref{g-constraint-est}.
We focus on the original form \eqref{obj0}--\eqref{g-constraint-est} in the remainder of this paper because
the form with the objective
function $f(\psi, \mu)$ can be rewritten in the form \eqref{obj0}--\eqref{g-constraint-est} by redefining 
$f$, $g_1$ and $g_2$.

\subsection{Examples of Empirical Applications}\label{subsec:examples} 

\subsubsection*{\underline{Example 1}}
\citet{BDM:98} 
use grouped data to estimate labor supply effects of tax reforms in the United Kingdom.
To motivate our setup, we consider a simple model with which \citet{BDM:98} 
explain how to use grouped data to estimate $\beta$ in the following labor supply model with no income effect:
\begin{align}\label{est-eq} 
h_{it} = \alpha + \beta \log w_{it} + U_{it}.
\end{align}
In this model,
$h_{it}$ and $w_{it}$, respectively, are hours of work and  the post-tax  hourly wage rate of individual $i$ in  year $t$,
and 
 $U_{it}$ is an unobserved random variable that satisfies certain conditions.
 The parameter $\beta$ is identified by a relation of the form
 $
 \beta = \beta( h_{gt},  lw_{gt}),
 $
 where $h_{gt}$ and $lw_{gt}$ are the mean hours and log wages in year $t$ of individuals in group $g$. 
 There are 8 groups defined by four year-of-birth cohorts and level of education. The data span the period 1978-1992.
 
 A nonparametric  version of \eqref{est-eq} is 
$h_{it} = \xi( w_{it}) + U_{it}$,
where $\xi \in \Xi$ is an unknown continuous function and $\Xi$ is a function space.
A nonparametric analog of $\beta$ is the weighted average derivative
\[
\tilde{\beta} = \int \frac{\partial \xi(u)}{\partial u} w(u) du,
\]
where $w$ is a non-negative weight function. 
The average derivative $\tilde{\beta}$ is not identified non-parametrically by the mean values of hours and wages for finitely many groups and time periods. It can be partially identified, however, by imposing a shape restriction such as weak monotonicity on the labor supply function $\xi$. 
Assume, for example, that $\mathbb{E}[h_{it} - \xi(w_{it}) | g,t] = 0$. 
{\citet{BDM:98} 
set  
$\mathbb{E}[h_{it} - \xi(w_{it}) | g,t] = a_g + m_t$, where $a_g$ and $m_t$, respectively, are group and time fixed effects. These are accommodated by our framework but
we do not do this in the present discussion.}

The identification interval for $\tilde{\beta}$ is $\tilde{\beta}_{-} \leq \tilde{\beta} \leq \tilde{\beta}_{+}$, where
\begin{align}\label{bdm-obj}
\tilde{\beta}_{+} = \max_{\xi \in \Xi}  \int \frac{\partial \xi(u)}{\partial u} w(u) du
\; \text{ and } \; 
\tilde{\beta}_{-} = \min_{\xi \in \Xi} \int \frac{\partial \xi(u)}{\partial u} w(u) du
\end{align}
subject to 
\begin{subequations}\label{g-constraint-bdm}
\begin{align}
\xi(w_{gt}) - \xi(w_{g't'}) &\leq 0 \;\; \text{if $w_{gt} < w_{g't'}$}, \label{g1-constraint-bdm} \\
h_{gt} - \xi(w_{gt}) &= 0. \label{g2-constraint-bdm}
\end{align}
\end{subequations}

The continuous mathematical programming problem \eqref{bdm-obj}-\eqref{g-constraint-bdm} can be put into the
finite-dimensional framework of \eqref{obj0}-\eqref{g-constraint-est} by observing that under mild conditions on $\Xi$, $\xi$ can be approximated very accurately by the truncated infinite series
\begin{align}\label{sieve-approx}
\xi(u) \approx \sum_{j=1}^{K}  \psi_j \phi_j(u),
\end{align}
where the $\psi_j$'s are constant parameters, the $\phi_j$'s are basis functions for $\Xi$, and $K$ is a truncation point. 
In an estimation setting, $K$ can be an increasing function of the sample size, though we do not undertake this extension here. 
 {The approximation error of \eqref{sieve-approx} can be bounded.  Here, however, we assume that $K$ is sufficiently large to make the error negligibly small.}
The finite-dimensional analog of \eqref{bdm-obj}-\eqref{g-constraint-bdm} is
\begin{align}\label{bdm-obj-fm}
J_{+} = \max_{\psi_j: j=1,\ldots,K} \sum_{j=1}^{K} \psi_j \int \frac{\partial \phi_j(u)}{\partial u} w(u) du
\; \text{ and } \; 
J_{-} = \min_{\psi_j: j=1,\ldots,K} \sum_{j=1}^{K} \psi_j \int \frac{\partial \phi_j(u)}{\partial u} w(u) du
\end{align}
subject to 
\begin{subequations}\label{g-constraint-bdm-fm}
\begin{align}
\sum_{j=1}^{K} \psi_j \left[ \phi_j (w_{gt}) - \phi_j (w_{g't'}) \right] &\leq 0 \;\; \text{if $w_{gt} < w_{g't'}$}, \label{g1-constraint-bdm-fm} \\
h_{gt} - \sum_{j=1}^{K} \psi_j  \phi_j (w_{gt})  &= 0. \label{g2-constraint-bdm-fm}
\end{align}
\end{subequations}
$J_{+}$ and $J_{-}$ can be estimated, thereby obtaining $\hat{J}_{+}$ and $\hat{J}_{-}$, by replacing $h_{gt}$ and $w_{gt}$ in \eqref{bdm-obj-fm}-\eqref{g-constraint-bdm-fm} with within-group sample averages and adding the constraint \eqref{S-constraint-est}.

\subsubsection*{\underline{Example 2}}

\citet[HP hereinafter]{Ho:Pakes:2014}  use the theory of revealed preference to develop an estimator of hospital choices by individuals. HP use data on privately insured births in California. We consider a simplified version of the HP model.

Using the notation of HP,
let  $p(c, h)$ denote the price an insurer is expected to pay at hospital $h$  for a patient with medical condition $c$.  Let  $i$ index patients.  Then $c_i$  is the medical condition of patient $i$, and $p(c_i, h)$  is the price an insurer is expected to pay at hospital $h$  for patient  $i$.
Let $l_i$ denote patient $i$'s location, $l_h$ hospital's location, and $d(\cdot,\cdot)$ the distance between the two locations. 
For hospitals $h \neq h'$, define 
\begin{align*}
\Delta p(c_i, h, h') := p(c_i, h) - p(c_i, h')
\; \text{ and } \;
\Delta d( l_{i}, l_{h}, l_{h'} ) := d(  l_{i}, l_{h} ) - d(  l_{i}, l_{h'} ).
\end{align*}
That is, $\Delta p(c_i, h, h')$ is the price difference between hospitals $h$ and $h'$ 
given patient condition $c_i$ 
and $\Delta d( l_{i}, l_{h}, l_{h'} )$ is the distance difference between hospitals $h$ and $h'$ 
given patient location $l_i$. 
Define 
\begin{align*}
u_{i, i'} (\psi, h, h') := 
\psi \left[ \Delta p(c_i, h, h') + \Delta p(c_{i'}, h', h) \right] 
-  \left[ \Delta d( l_{i}, l_{h}, l_{h'} ) + \Delta d( l_{i'}, l_{h'}, l_{h} ) \right],
\end{align*}
where $\psi$ is a scalar parameter that determines price sensitivity relative to distance.
Note that the coefficient for distance is normalized to be $-1$.
$\psi$ is the key parameter in HP.

Define the four-dimensional vector of instruments based on distance:
\begin{align*}
z_{i, i'} (h, h') := 
\begin{pmatrix}
\max\{ \Delta d( l_{i}, l_{h}, l_{h'} ), 0\} \\
- \min \{ \Delta d( l_{i}, l_{h}, l_{h'} ), 0\} \\
\max\{ \Delta d( l_{i'}, l_{h'}, l_{h} ), 0\} \\
- \min \{ \Delta d( l_{i'}, l_{h'}, l_{h'} ), 0\} 
\end{pmatrix}.
\end{align*} 
Here, the instruments are based on distance measures
and constructed to be positive to preserve the signs of the inequalities below
in \eqref{HP-identifying}.

Let $S(h,h',s)$ be the set of patients with severity $s$ who chose hospital $h$ but had hospital $h'$ in their choice set. 
The identifying assumption in HP is that 
\begin{align}\label{HP-identifying}
\mathbb{E} \left[ 
u_{i, i'} (\psi, h, h') z_{i, i'} (h, h')
\Big |
i \in S(h,h',s), i' \in S(h',h,s)
\right] \geq 0
\end{align}
for all $s, h, h'$ such that $h \neq h'$. 
We can rewrite \eqref{HP-identifying} as
\begin{align*}
\psi \mu_p (h,h',s) - \mu_d(h,h',s) \geq 0 \; \text{ for all $(h,h',s)$ such that $h \neq h'$},
\end{align*}
where 
\begin{align*}
\mu_p (h,h',s) &:= 
\mathbb{E} \left[ 
z_{i, i'} (h, h') \{
\Delta p(c_i, h, h') + \Delta p(c_{i'}, h', h)  \} 
\Big |
i \in S(h,h',s), i' \in S(h',h,s) \right], \\ 
\mu_d (h,h',s) &:=
\mathbb{E} \left[ 
z_{i, i'} (h, h') \{
\Delta d( l_{i}, l_{h}, l_{h'} ) + \Delta d( l_{i'}, l_{h'}, l_{h} )  \} 
\Big |
i \in S(h,h',s), i' \in S(h',h,s) \right].
\end{align*}
To see the connection between our general framework and HP's inequality estimator, let 
$f(\psi) = \psi$, $\mu = (\mu_p, \mu_d)$, and $g_1$ be a collection of inequalities such that  
\begin{align*}
g_1(\psi, \mu) =  \mu_d(h,h',s) - \psi \mu_p (h,h',s) \leq 0  \; \text{ for all $(h,h',s)$ such that $h \neq h'$}.
\end{align*}
There is no element in $g_2$ (no equality constraints here). 
Since each element in $\mu$ can be estimated by a suitable sample mean,  our general framework includes HP's estimator as a special case.

\subsection{Analysis}\label{subsec:analysis}

This section presents our three methods for forming finite-sample lower bounds on
$$\mathbb{P} \left[ \hat{J}_{-}(\bar{X}) \leq J_{-} \leq f(\psi) \leq J_{+} \leq \hat{J}_{+}(\bar{X}) \right].$$
The three bounds make assumptions of differing strengths about the distribution of $X$. The bounds are tighter in samples of moderate size with stronger assumptions. All proofs are in Online Appendix~\ref{app:proof}. We begin with the following theorem, which applies to all the methods and forms the basis of our approach.

\begin{theorem}\label{thm2}
Assume that $g_1(\psi, \mu) \leq 0$ and $g_2(\psi, \mu) = 0$ for some $\psi$. Then
\begin{align}\label{main-result}
\mathbb{P} \left[ \hat{J}_{-}(\bar{X}) \leq J_{-} \leq f(\psi) \leq J_{+} \leq \hat{J}_{+}(\bar{X}) \right] \geq \mathbb{P} \left[ n^{1/2} (\bar{X} - \mu) \in \mathcal{S} \right].
\end{align}
\end{theorem}

Now define
\begin{align*}
Z_i := X_i - \mu \;\; \text{ and } \;\; 
\bar{Z} := n^{-1/2} \sum_{i=1}^n Z_i = n^{-1/2} \sum_{i=1}^n (X_i - \mu).
\end{align*}
Then $\mathbb{E}( \bar{Z} ) = 0$. Note that  $\Sigma = 
\text{cov}(Z_i) = \text{cov}(\bar{Z}) = \text{cov}(X)$. 
 We make the following assumption throughout the remainder of the paper.

\begin{assumption}\label{a1}
(i) $\{ X_i : i=1,\ldots,n \}$ is an independent random sample from the distribution of $X$.
(ii) $\mathcal{S}$ is compact and convex.
(iii) $\Psi$ is compact. 
(iv) $f(\psi)$ is bounded on $\Psi$.
(v) $g_1(\psi, \mu) \leq 0$ and $g_2(\psi, \mu) = 0$ for some $\psi$.
\end{assumption}

\subsubsection{Known $\Sigma$}\label{sec:known:sigma}

Suppose for the moment that $\Sigma$ is known. 
Section~\ref{sec:unknown:sigma} 
discusses the case in which $\Sigma$ is unknown. 
If  $\Sigma$ is non-singular, 
let $[ \Sigma^{-1/2} (X_i - \mu) ]_j$ denote the $j$'th component of $\Sigma^{-1/2} (X_i - \mu)$.

\subsubsection*{\underline{Method 1}}

This method approximates the distribution of $\bar{Z}$ by a normal distribution.  
To do this, make the following assumption.

\begin{assumption}\label{a2}
(i) $\Sigma$ is non-singular, and its components are all finite.
(ii) There is a constant $\overline{\mu}_3 < \infty$ such that
$\mathbb{E} \left( \abs{ [ \Sigma^{-1/2} (X_i - \mu) ]_j }^3 \right) \leq \overline{\mu}_3$
 for all $i=1,\ldots,n$ and $j=1,\ldots,p$. 
\end{assumption}

Define 
the independent random $p$-vectors
$W_i \sim N(0, \Sigma)$ $(i=1,\ldots,n)$ and
$\bar{W} := n^{-1/2} \sum_{i=1}^n W_i \sim N(0, \Sigma)$.
The multivariate generalization of the
Lindeberg-L\'{e}vy central limit theorem shows that $\bar{Z}$ is asymptotically distributed as
$N(0, \Sigma)$, so the distribution of $\bar{Z}$ can be approximated by that of $W$.
The following theorem bounds the error of this approximation.

\begin{theorem}\label{lem1}
Let Assumptions \ref{a1}(i),(ii) and \ref{a2} hold. Then
\[
\abs{ \mathbb{P} (\bar{Z} \in \mathcal{S}) - \mathbb{P} (W \in \mathcal{S}) }
\leq \frac{(42p^{1/4} + 16) p^{3/2} \overline{\mu}_3}{ n^{1/2}}.
\] 
\end{theorem}

Theorem~\ref{lem1} approximates the distribution of  $\bar{Z}$ by a multivariate normal distribution and uses a multivariate generalization of the Berry-Ess\'{e}en theorem \citep{Raic2019} to bound the approximation error.  Theorem~\ref{lem1} implies that for any $0 < \alpha < 1$,
\begin{align}\label{CI-chi2}
\mathbb{P} \left\{ n (\bar{X} - \mu)' \Sigma^{-1} (\bar{X} - \mu) \leq \kappa_{\chi^2_{p}}(1-\alpha)  \right\}
\geq (1-\alpha) - \textrm{B}(n,p,\overline{\mu}_3),
\end{align}
where $\kappa_{\chi^2_{p}}(1-\alpha)$ is the $(1-\alpha)$ quantile of the chi-square distribution with $p$ degrees of freedom
and
\begin{align*}
\textrm{B}(n,p,\overline{\mu}_3)  :=  \frac{(42p^{1/4} + 16) p^{3/2} \overline{\mu}_3}{ n^{1/2}}.
\end{align*}
The term $\textrm{B}(n,p,\overline{\mu}_3)$ in \eqref{CI-chi2}  is asymptotically negligible but can be large
in samples of moderate
size because it accommodates ``worst case'' distributions of $\bar{X}$ that may be far from normal. 
If $n$ is large enough that $\textrm{B}(n,p,\overline{\mu}_3) < \alpha $,
then it follows from  \eqref{CI-chi2} that 
\begin{align}\label{CI-chi2-finite}
\mathbb{P} \left\{ n (\bar{X} - \mu)' \Sigma^{-1} (\bar{X} - \mu) \leq \kappa_{\chi^2_{p}} \big[ 1-\alpha + \textrm{B}(n,p,\overline{\mu}_3) \big]  \right\}
\geq 1-\alpha.
\end{align}
It follows from Theorem~\ref{thm2} that \eqref{CI-chi2} and \eqref{CI-chi2-finite} provide lower bounds on the coverage probabilities of confidence intervals for $f(\psi)$ when $\mathcal{S}$ is an ellipsoid.


Table~\ref{tab:bentkus} shows numerical values of 
the  bound $(1-\alpha) - \textrm{B}(n,p,\overline{\mu}_3)$ 
and critical value $\kappa_{\chi^2_{p}} \big[ 1-\alpha + \textrm{B}(n,p,\overline{\mu}_3) \big]$
for different values of $n$ and $p$ at $\alpha = 0.05$
and $\overline{\mu}_3 = 2$.
To have a bound close to $1-\alpha$ and a finite critical value, $n$ must be very large, especially if $p$ is large. This
is because \eqref{CI-chi2}  accommodates worst case distributions of $\bar{X}$.
Methods 2 and 3, which are discussed next in this section, provide tighter bounds and smaller critical values when $n$ is smaller, though Method 1 can provide a smaller critical value when $n$ is very large and $p$ is small.
However, Method 1 is hard to use in applications even when $n$ is large 
if $\Sigma$ and $\overline{\mu}_3$ are
unknown, because the resulting bounds depend on population parameters that are difficult to estimate. This problem is discussed in Section~\ref{sec:unknown:sigma}.

\begin{table}[htbp]
\centering
\caption{Values of the Bound  and Critical Value}
\label{tab:bentkus}
\medskip
\begin{tabular}{rlllll}
  \hline
 & $p=1$ & 2 & 3 & 4 & 5 \\ 
  \hline
$n=10^5$ & 0.58 ($\infty$) & 0.00 ($\infty$) & 0.00 ($\infty$) & 0.00 ($\infty$) & 0.00 ($\infty$) \\ 
$10^6$ & 0.83 ($\infty$) & 0.58 ($\infty$) & 0.21 ($\infty$) & 0.00 ($\infty$) & 0.00 ($\infty$) \\ 
$10^7$ & 0.91  (6.13) & 0.83 ($\infty$) & 0.72 ($\infty$) & 0.57 ($\infty$) & 0.39 ($\infty$) \\ 
$10^8$ & 0.94 (4.29) & 0.91 (8.73) & 0.88 ($\infty$) & 0.83 ($\infty$) & 0.77 ($\infty$) \\ 
$10^9$ & 0.95 (3.97) & 0.94 (6.53) & 0.93 (9.21) & 0.91 (12.88) & 0.89 ($\infty$) \\ 
   \hline
\end{tabular}

\medskip

\begin{minipage}{0.8\textwidth} 
{Note. 
The table shows the values of the bound
$\max[ 0, (1-\alpha) - \textrm{B}(n,p,\overline{\mu}_3)]$
and critical value $\kappa_{\chi^2_{p}} \big[ 1-\alpha + \textrm{B}(n,p,\overline{\mu}_3) \big]$
(in parentheses)  for different  $n$ and $p$
with $\alpha = 0.05$ and $\overline{\mu}_3 = 2$.
\par}
\end{minipage}
\end{table}

\subsubsection*{\underline{Method 2}}

Method 2 obtains bounds that are much tighter than those of Method 1 when $n$ is smaller than in Table~\ref{tab:bentkus}, and Method 2 does not require $\Sigma$ to be invertible. This is accomplished by assuming that the distributions of the components of $X$ are sub-Gaussian. Specifically, make the following assumption.

\begin{assumption}\label{a4}
Let $\Upsilon$ denote a non-singular, non-stochastic $p \times p$ matrix with finite elements.  
For any such $\Upsilon$; 
all $i=1,\ldots,n$; and all $j=1,\ldots,p$;
there are finite constants $\sigma_j^2 > 0$ such that the distribution of 
 $\tilde{Z}_{ij} := [ \Upsilon^{-1/2} (X_i - \mu) ]_j$ is sub-Gaussian with variance proxy $\sigma_j^2$. That is, $\mathbb{E} [ \exp( \lambda \tilde{Z}_{ij} ) ] \leq \exp (\sigma_j^2 \lambda^2 / 2)$ for all $\lambda \in \mathbb{R}$; $i=1,\ldots,n$; and $j=1,\ldots,p$.
\end{assumption}

Assumption~\ref{a4} requires that the distribution of $\tilde{Z}_{ij}$ be thin-tailed. 
Sub-Gaussian random variables include Gaussian, Rademacher, and bounded random variables as special cases. See,  e.g., \citet{wainwright2019book}.
There is a tradeoff between $\Upsilon$ and $\sigma_j^2$.
In particular, $\sigma_j^2$ may be larger if $\Upsilon$ is the $p \times p$ identity matrix than if $\Upsilon = \Sigma$ and $\Sigma$ is non-singular. 

Define  $\sigma^2 := \max_{1 \leq j \leq p} \sigma_j^2$.
The following theorem, combined with Theorem~\ref{thm2}, provides  a lower bound on a confidence interval for $f(\psi)$ when $\mathcal{S}$ is an ellipsoid.

\begin{theorem}\label{lem-SG-1}
Let Assumptions \ref{a1}(i) and \ref{a4} hold. Then,
for any $t > 0$, 
\begin{align*}
\mathbb{P} \left\{ n (\bar{X} - \mu)' \Upsilon^{-1} (\bar{X} - \mu) > t \right\}
\leq 
2\sum_{j=1}^p 
\exp \left( - \frac{t}{2 p \sigma_j^2} \right).
\end{align*}
In addition,  for any $0 < \alpha < 1$,
\begin{align*}
\mathbb{P} \left\{ \frac{n}{\sigma^2} (\bar{X} - \mu)' \Upsilon^{-1} (\bar{X} - \mu) \leq \kappa_{SG, 1}(1-\alpha)  \right\}
\geq (1-\alpha),
\end{align*}
where 
$\kappa_{SG, 1}(1-\alpha) := 2 p \cdot  \log (2p/\alpha)$.
\end{theorem}

Theorem~\ref{lem-SG-1} and Method 2 make use of the sub-Gaussianity of $\tilde{Z}_{ij}$, whereas Theorem~\ref{lem1} and Method 1
allow the tails of the distribution of $X$ to be thicker than sub-Gaussian tails. The
critical values of Methods 1 and 2 are compared later in this section after the description of Method
3. Estimation of $\sigma^2$ is discussed in Section~\ref{sec:unknown:sigma}.

\subsubsection*{\underline{Method 3}}

Method 3 makes the stronger assumption that the distribution of $X$ is sub-Gaussian in a vector sense. Specifically, Method 3 makes the following assumption.

\begin{assumption}\label{a5}
Let $\Upsilon$ denote a non-singular $p \times p$ matrix with finite elements.  
There is a finite constant $\sigma^2 > 0$ such that 
\begin{align*}
\mathbb{E} \left[ \exp ( \lambda'  \Upsilon^{-1/2} \bar{Z}) \right] 
\leq \exp ( \lambda' \lambda \sigma^2 / 2)
\end{align*}
for all $\lambda \in \mathbb{R}^p$, where
$\bar{Z} = n^{-1/2} \sum_{i=1}^n (X_i - \mu)$.
\end{assumption}

Assumption~\ref{a5} is stronger than Assumption~\ref{a4}, because Assumption~\ref{a5} requires the entire vector $X$ to be sub-Gaussian.  
If $X$ is multivariate normal and $\Upsilon = \Sigma$, 
Assumption~\ref{a5} holds with $\sigma^2 = 1$. 
In general, however, it is difficult to find simple conditions under which
Assumption~\ref{a5} is satisfied
without assuming that the elements of $\bar{Z}$ are independent of one another.

An application of Theorem 2.1 of \citet{HKT:2021} gives the following theorem
which, combined with Theorem~\ref{lem1}, provides a lower bound on the coverage probability of a confidence interval for $f(\psi)$ when $\mathcal{S}$ is an ellipsoid.

\begin{theorem}\label{lem-SG-2}
Let Assumptions \ref{a1}(i) and \ref{a5} hold. 
 Then, for any $0 < \alpha < 1$,
\begin{align*}
\mathbb{P} \left\{ \frac{n}{\sigma^2} (\bar{X} - \mu)' \Upsilon^{-1} (\bar{X} - \mu) \leq \kappa_{SG, 2}(1-\alpha)  \right\}
\geq 1-\alpha,
\end{align*}
where 
$\kappa_{SG, 2}(1-\alpha) := p + 2 \sqrt{ p \cdot \log (1 / \alpha)} + 2 \log (1 / \alpha)$.
\end{theorem}

This theorem and Method 3 
yield a smaller critical value and confidence set than Theorem~\ref{lem-SG-1} and Method 2 do, but they require the stronger Assumption 4. Table~\ref{tab:sg-cv} shows the critical values of Methods 2 and 3 and chi-square critical values with p degrees of freedom. The critical
values are all for $\alpha = 0.05$. None of the critical values depends on $n$. The chi-square critical value achieves an asymptotic coverage probability of 0.95 and yields the smallest confidence set but does not ensure a finite-sample coverage probability of at least 0.95. The sub-Gaussian critical values and resulting confidence sets are larger but ensure finite-sample coverage probabilities of at least 0.95 under the regularity conditions of Theorems~\ref{lem-SG-1} and ~\ref{lem-SG-2}. Assumption~\ref{a4} is easier to satisfy, and its variance proxy is easier to estimate in applications, but the critical value of Method 2 increases more rapidly than the critical value of Method 3 as $p$ gets large.

\begin{table}[htbp]
\centering
\caption{Different Critical Values}
\label{tab:sg-cv}
\medskip

\begin{tabular}{rrrr}
  \hline
$p$  & $\kappa_{\chi^2_{p}}(1-\alpha)$ & $\kappa_{SG, 1}(1-\alpha)$ & $\kappa_{SG, 2}(1-\alpha)$ \\ 
  \hline
 1 & 3.84 & 7.38 & 10.45 \\ 
 2 & 5.99 & 17.53 & 12.89 \\ 
 3 & 7.81 & 28.72 & 14.99 \\ 
 4 & 9.49 & 40.60 & 16.91 \\ 
 5 & 11.07 & 52.98 & 18.73 \\ 
   \hline
\end{tabular}

\medskip

\begin{minipage}{0.6\textwidth} 
{Note. The table shows different critical values at $\alpha = 0.05$.
$\kappa_{\chi^2_{p}}(1-\alpha)$ is the chi-square critical value with $p$ degrees of freedom;
$\kappa_{SG, 1}(1-\alpha) = 2 p \cdot  \log (2p/\alpha)$;
and
$\kappa_{SG, 2}(1-\alpha) = p + 2 \sqrt{ p \cdot \log (1 / \alpha)} + 2 \log (1 / \alpha)$.
\par}
\end{minipage}
\end{table}

The critical values of Methods 2 and 3 in Table~\ref{tab:sg-cv} can be compared with those of Method 1 in Table~\ref{tab:bentkus}. The critical value of Method 1 converges to the chi-square critical value as $n \rightarrow \infty$. The critical values of Methods 2 and 3 do not depend on $n$. Therefore, the critical value of Method 1 is smaller than those of Methods 2 and 3 when $n$ is very large and $p$ is small enough. However,
the critical value of Method 1 is infinite with moderate values of $n$, whereas the critical values of Methods 2 and 3 are finite at all values of $n$.

\subsubsection{Unknown $\Sigma$ and sub-Gaussian variance proxy $\sigma^2$}\label{sec:unknown:sigma} 

In applications, $\Sigma$ and the sub-Gaussian variance proxy $\sigma^2$ 
are unknown except in special cases. This section explains how to estimate these quantities and discusses the effect of estimation on the bounds presented in 
Section \ref{sec:known:sigma}.

We begin with the bound of Theorem~\ref{lem1} and~\eqref{CI-chi2}.
Let $\widehat{\Sigma}$ be the following estimator of $\Sigma$:
\begin{align*}
\widehat{\Sigma} 
&:= n^{-1} \sum_{i=1}^n X_i X_i' - \bar{X} \bar{X}'.
\end{align*}
Let $\Sigma_{jk}^{-1}$ denote the $(j,k)$ component of $\Sigma^{-1}$.
For each $j,k = 1,\ldots,p$, let
\begin{align*}
\Sigma_{jj} := \mathbb{E} \left[ ( X_{ij} - \mu_j )^2 \right]
\; \text{ and } \;
\Gamma_{jk} := \mathbb{E} \left[ \{ (X_{ij} - \mu_j)(X_{ik} - \mu_k) - \Sigma_{jk} \}^2 \right]. 
\end{align*}

\begin{assumption}\label{a3}
(i) There is a constant $C_\Sigma < \infty$  such that $\abs{ \Sigma_{jk}^{-1} } \leq C_\Sigma$ for each $j,k = 1,\ldots,p$. 
(ii) There is a finite constant $\kappa_0 \geq 1$ such that
\begin{align}\label{kappa0}
\Sigma_{jj} \leq \kappa_0
\ \ \text{ and } \ \
\Gamma_{jk} \leq \kappa_0
\end{align}
for every $j,k = 1,\ldots,p$. 
(iii)  There is a finite constant $\kappa_1$ such that 
\begin{subequations}\label{kappa1-min}
\begin{align}
\mathbb{E} \left[ \abs{ X_{ij} - \mu_j }^r \right] &\leq \kappa_1^{r-2} \frac{r!}{2} \Sigma_{jj}, \label{kappa1-min-1} \\
\mathbb{E} \left[ \abs{ (X_{ij} - \mu_j)(X_{ik} - \mu_k) - \Sigma_{jk} }^r \right] &\leq \kappa_1^{r-2} \frac{r!}{2} \Gamma_{jk} \label{kappa1-min-2} 
\end{align}
\end{subequations}
for every $r=2,3,\ldots$ and $j,k = 1,\ldots,p$.
\end{assumption}

Assumption~\ref{a3} is stronger than 
Assumption~\ref{a2}.
In particular, \eqref{kappa1-min-2} 
implies that the $(X_{ij} - \mu_j)(X_{ik} - \mu_k)$ is sub-exponential.  Therefore, $X_{ij}$ is sub-Gaussian  because
a random
variable is sub-Gaussian if and only if its square is sub-exponential 
\citep[][Lemma 2.7.6]{vershynin2018high}.
Also, the product of two sub-Gaussian variables is sub-exponential
\citep[][Lemma 2.7.7]{vershynin2018high}.
We use Assumption \ref{a3}(iii) to apply Bernstein's inequality 
to the bound of Method 1
\citep[see, e.g., Lemma 14.13 of][]{buhlmann2011statistics}.

Define the random vector  
$\widehat{W} \sim N(0, \widehat{\Sigma})$.
We approximate the distribution of $W$ by the distribution of  
$\widehat{W}$ with $\widehat{\Sigma}$ treated as a non-stochastic matrix. 
Define
$\mathbf{P} (\mathcal{S}, \Sigma) := \mathbb{P} (W \in \mathcal{S})$ for $W \sim N(0, \Sigma)$
and
\begin{align}\label{choice:rt:final}
r_n(t) & :=
  8  \sqrt{\frac{2 \kappa_0 t}{n} }, \
w_n(t) := C_\Sigma p^3 2^{p+1} r_n(t),
\ \text{ and } \
\kappa^* := \min \left\{ (4\kappa_1)^{-1}, 2 \kappa_1^{-2}, (32 \kappa_0)^{-1} \right\}.
\end{align} 
The following lemma gives a finite-sample bound on the error of the approximation.

\begin{lemma}\label{lem2}
Let Assumptions \ref{a1}, \ref{a2}, and \ref{a3} hold and 
\begin{align}\label{tech:condition:further}
\log(2p) \leq  t \leq \kappa^* n.
\end{align}
Then, 
\[
\abs{   \mathbf{P} (\mathcal{S}, \widehat{\Sigma}) - \mathbb{P} (W \in \mathcal{S}) }
\leq w_n(t)
\] 
with probability at least $1- 2 e^{-t}$.
\end{lemma}

The condition \eqref{tech:condition:further} is a mild technical condition that can be satisfied easily. The conclusion of Lemma~\ref{lem2} holds only if $\widehat{\Sigma}$ satisfies certain conditions that are stated in the
  proof of the lemma in Online Appendix \ref{app:proof}. These conditions are satisfied with probability at least 
$1-2 e^{-t}$, not with certainty.

Define
\begin{align}
\delta_n^* := \min_t \left[ w_n(t) + 2 e^{-t} \right]  \  \text{ subject to \eqref{tech:condition:further}}.
\end{align}
Now combine Theorem~\ref{lem1} and Lemma~\ref{lem2} to obtain the following theorem.

\begin{theorem}\label{thm1}
Let Assumptions \ref{a1}, \ref{a2}, \ref{a3} and \eqref{tech:condition:further} hold. 
Then,
\[
\abs{   \mathbb{P} (\bar{Z} \in \mathcal{S}) - \mathbf{P} (\mathcal{S}, \widehat{\Sigma}) }
\leq \frac{(42p^{1/4} + 16) p^{3/2} \overline{\mu}_3}{ n^{1/2}} + \delta_n^*.
\] 
\end{theorem}

Theorem \ref{thm1} provides a finite-sample upper bound on the error made by approximating 
$\mathbb{P} \left[ n^{1/2} (\bar{X} - \mu) \in \mathcal{S} \right]$
by 
$\mathbf{P} (\mathcal{S}, \widehat{\Sigma})$. 
Combining Theorems \ref{thm2} and \ref{thm1} yields

\begin{theorem}\label{thm3}
Let Assumptions \ref{a1}, \ref{a2}, \ref{a3} and \eqref{tech:condition:further} hold.  
Then, 
\begin{align}\label{main-result-final}
&\mathbb{P} \left[ \hat{J}_{-}(\bar{X}) \leq J_{-} \leq f(\psi) \leq J_{+} \leq \hat{J}_{+}(\bar{X}) \right] 
\geq 
\mathbf{P} (\mathcal{S}, \widehat{\Sigma})
- \left\{ \frac{(42p^{1/4} + 16) p^{3/2} \overline{\mu}_3}{ n^{1/2}} + \delta_n^* \right\}.
\end{align}
\end{theorem}

Theorem \ref{thm3} provides a finite-sample lower bound on 
$\mathbb{P} \left[ \hat{J}_{-}(\bar{X}) \leq J_{-} \leq f(\psi) \leq J_{+} \leq \hat{J}_{+}(\bar{X}) \right]$
that takes account of random sampling error in $\hat{\Sigma}$. It is not
difficult to choose $\mathcal{S}$ so that the right-hand side of 
\eqref{main-result-final} is $1-\alpha$ for any $0 \leq \alpha \leq 1$ if $n$ is large enough and $p$ is small enough to make the term in square brackets on right-hand side of the
inequality less than $\alpha$. However, the presence of $\delta_n^*$ greatly decreases the right-hand side of \eqref{main-result-final}
relative to what it is when $\Sigma$ is known, thereby increasing the size of the confidence region $\mathcal{S}$ for any given value of $\alpha$. In addition, the right-hand side of \eqref{main-result-final} depends on several population parameters that are difficult to estimate. Therefore, \eqref{main-result-final} is of limited use for applications. The bounds for Methods 2 and 3 with an unknown variance proxy also depend on unknown population parameters, but these are easier to estimate, as is discussed in Section~\ref{sec:methods2and3}. The dependence of finite-sample bounds on unknown population parameters is unavoidable except in special cases.
For example, if $\Upsilon = I_p$ in Theorem~\ref{lem-SG-1} and each element of $X-\mu$ is contained in $[-1, 1]$, 
then $\sigma_j^2 = 1$ for all $j$, and the first inequality in Theorem~\ref{lem-SG-1} becomes the Hoeffding inequality.

We now consider Method 2, which is the sub-Gaussian case of Assumption~\ref{a4}. In some special cases, the sub-Gaussian variance proxies, $\sigma_j^2$ are known.
For example, if $\Upsilon = I_p$ in Theorem~\ref{lem-SG-1}
and each element of $X - \mu$ is contained in $[-1,1]$, then $\sigma_j^2 = 1$ for all $j$. 
Here, we derive bounds for the case in which 
$\sigma^2 = \max_{1 \leq j \leq p} \sigma_j^2$ is unknown and must be estimated.

For all $i=1,\ldots,n$ and $j=1,\ldots,p$,
let
\begin{align*}
\tilde{X}_{ij} := [\Upsilon^{-1/2} X_{i} ]_j,
\end{align*}
and the sample variance of $\tilde{X}_{ij}$
\begin{align*}
\tilde{\Sigma}_{jj} := n^{-1} \sum_{i=1}^n  \tilde{X}_{ij}^2 - \left( n^{-1} \sum_{i=1}^n \tilde{X}_{ij} \right)^2.
\end{align*}
The following lemma establishes a finite-sample probability bound on the absolute difference between the sample and population variances of $\tilde{X}_{ij}$.
Recall that $\tilde{Z}_{ij} = [ \Upsilon^{-1/2} (X_i - \mu) ]_j$  defined in Assumption~\ref{a4}.

\begin{lemma}\label{lem:sample-var}
Let Assumptions \ref{a1}(i) and \ref{a4} hold. Then, for any $0 < t \leq (16 \sigma_j^2)^2 \, n$, 
\begin{align*}
& \mathbb{P} \left[ 
\left| \tilde{\Sigma}_{jj} - \mathbb{E} [ \tilde{Z}_{ij}^2 ] \right| 
> 
3 \left( \frac{ t}{n} \right)  + 
\left( 1+ 2 \left| n^{-1} \sum_{i=1}^n \tilde{X}_{ij} \right| \right) \left( \frac{ t}{n} \right)^{1/2} 
 \right] \leq
 e^{- t / (512 \sigma_j^4) }  + 2 e^{-t / (2 \sigma_j^2)}.
\end{align*}
\end{lemma}

The next lemma establishes a link between the population variance of
$\tilde{Z}_{ij}$
to the variance proxy $\sigma_j^2$.

\begin{lemma}\label{lem:sample-var-proxy-equ}
Let Assumptions \ref{a1}(i) and \ref{a4} hold. 
Then, the population variance $\mathbb{E} [ \tilde{Z}_{ij}^2 ]$ is a lower bound on the variance proxy $\sigma_j^2$.
\end{lemma}

Now define
\begin{align}\label{est-SG-parameter}
\hat{\sigma}^2 (t)
:=  \max_{j=1,\ldots, p} 
\left\{ 
\tilde{\Sigma}_{jj}
+ 3 \left( \frac{ t}{n} \right)  + 
\left( 1+ 2 \left| n^{-1} \sum_{i=1}^n \tilde{X}_{ij} \right| \right) \left( \frac{ t}{n} \right)^{1/2} 
\right\},
\end{align} 
which we use to estimate  $\sigma^2 = \max_{1 \leq j \leq p} \sigma_j^2$.
  The following theorem holds for the sub-Gaussian
case of Assumption~\ref{a4} (Method 2).

 \begin{theorem}\label{thm3-SG}
Let Assumptions \ref{a1}(i),(iii),(iv),(v), and \ref{a4} hold. 
Let the restriction $n^{1/2} (\bar{X} - \mu) \in \mathcal{S}$ be given by
\begin{align}\label{main-result-final-SG-CI}
n (\bar{X} - \mu)' \Upsilon^{-1} (\bar{X} - \mu)
\leq 2 \hat{\sigma}^2 (t) \cdot p \cdot  \log (2p/\alpha),
\end{align}
where $\Upsilon$ is non-random and 
and $\hat{\sigma}^2 (t)$ is 
given in \eqref{est-SG-parameter}.
 Then, 
for any $0 < t \leq (16 \sigma_j^2)^2 \, n$, 
\begin{align}\label{main-result-final-SG}
\mathbb{P} \left[ \hat{J}_{-}(\bar{X}) \leq J_{-} \leq f(\psi) \leq J_{+} \leq \hat{J}_{+}(\bar{X}) \right] 
&
\geq 1 - \alpha - 
\sum_{j=1}^p \left( e^{- t / (512 \sigma_j^4) }  + 2 e^{-t / (2 \sigma_j^2)} \right).
\end{align}
\end{theorem}

A smaller choice of $t$ in \eqref{est-SG-parameter} makes the confidence set in \eqref{main-result-final-SG-CI} tighter but results in a
lower probability bound \eqref{main-result-final-SG}. The right-hand side of \eqref{main-result-final-SG} can be close to zero and depends on the unknown parameters $\sigma_j^2$. Therefore, the bound \eqref{main-result-final-SG}, like the bound \eqref{main-result-final}, is of limited use
for applications. As is noted in the discussion of \eqref{main-result-final}, dependence of finite-sample bounds on unknown population quantities is unavoidable except in special cases. Section \ref{sec:methods2and3} describes a practical approach to dealing with this problem that can be used in applications.

A result similar to Theorem \ref{thm3-SG} can be obtained for Method 3. We do not undertake this here, however, because the resulting bound analogous to \eqref{main-result-final-SG} is loose, depends on unknown population parameters, and is of limited use in applications. Instead, in Section \ref{sec:methods2and3}, we describe a way of dealing with unknown population parameters in Methods 2 and 3 that can be used in applications.

\subsection{Dealing with Unknown Population Parameters in Methods 2 and 3 in Applications}\label{sec:methods2and3}

Except in special cases, it is not possible to obtain finite-sample inequalities for Methods 2 and 3 that do not depend on unknown population parameters. It is possible, however, to estimate lower bounds
on these parameters consistently. It follows from
Lemma~\ref{lem:sample-var-proxy-equ}
that 
$\tilde{\Sigma}_{jj}$ and $\max_{1 \leq j \leq p} \tilde{\Sigma}_{jj}$, respectively, are consistent estimates of lower bounds on the variance proxies $\sigma_j^2$ and $\sigma^2$ of Method 2. 
The differences between the lower bounds and the variance proxies are often small.
Arguments like those used to prove Lemma~\ref{lem:sample-var-proxy-equ} show that the
largest eigenvalue of the covariance matrix of 
$\tilde{Z}_{ij} (j =1,\ldots,p)$ is a lower bound on the variance proxy of Method 3.
This can be estimated consistently by the largest eigenvalue of the sample covariance matrix. 

Standard methods can be used to obtain asymptotic confidence intervals for the variance proxies of Method 2. These methods do not provide information about differences between true and nominal coverage probabilities in finite samples, but they provide practical indications of the magnitudes of the Method 2 bounds that can be implemented in applications. 

Obtaining a useful asymptotic confidence interval for the largest eigenvalue of the covariance matrix of $\tilde{Z}_{ij}$ is difficult. A wide interval whose true coverage probability is likely to
be much greater than the nominal probability can be obtained from the Frobenius norm of the difference between the estimated and true covariance matrices.

Table~\ref{MC:tab} in Section~\ref{sec:MC:new} compares values of $\hat{J}_{+}$ obtained using estimates and population values of the variance proxies of Methods 2 and 3.

\section{Computational Algorithms}\label{sec:algorithms}

Recall that our general framework is to obtain the bound 
$[\min_{\psi, m} f(\psi), \max_{\psi, m} f(\psi)]$
subject to 
$g_1(\psi, m) \leq 0,  g_2(\psi, m) = 0, \psi \in \Psi,$
and
$n^{1/2} (\bar{X} - m) \in \mathcal{S}$.
%
%
 In many examples, $f(\psi)$ is linear in $\psi$. For example, $\psi$ is the vector of all the parameters in an econometric model and $f(\psi)$ is just one element of $\psi$ or a linear combination of elements of $\psi$. 


The restrictions $g_1(\psi, \mu) \leq 0$ include shape restrictions among the elements of $\psi$. Equality restrictions are imposed via $g_2(\psi, \mu) = 0$. The easiest case is that $g_j(\psi, \mu)$ is linear in $(\psi, \mu)$ for each $j=1,2$. In some of examples we consider, $g_j(\psi, \mu)$ is linear in $\psi$, holding $\mu$ fixed, and 
linear in $\mu$, keeping $\psi$ fixed, but not linear  in $(\psi, \mu)$ jointly. This corresponds to the case of bilinear constraints. For example, $g_j(\psi, \mu)$ may depend on the product between one of elements of $\psi$ and one of elements of $\mu$.  
In practice, $\Psi$ can always be chosen large enough that the constraint $\psi \in \Psi$ is not binding additionally and can be ignored.
 For example, suppose that $\psi$ is a probability and the constraints in $g_1(\psi, \mu) \leq 0$ and $g_2(\psi, \mu) = 0$ impose a restriction on $\psi$ such as $[a, b]$
for $0 \leq a < b \leq 1$. Then, it is not necessary to impose $\psi \in \Psi = [0,1]$ additionally. 


Recall that the leading cases of $\mathcal{S}$ include an ellipsoid and a box. 
For brevity, we focus on the scenario that normal distributions are used in obtaining $\mathcal{S}$.  
When  $\mathcal{S}$ is a box, the critical value  $\kappa_b(1-\alpha)$ can be easily simulated  from the $N(0, \widehat{\Sigma})$ and the restriction
$n^{1/2}(\bar{X} - \mu) \in \mathcal{S}$
can be written as linear constraints. 
When $\mathcal{S}$ is an ellipsoid,  the critical value $\kappa_e(1-\alpha)$ can be obtained from the $\chi^2( d_\mu )$ distribution, where 
$d_\mu$ is the dimension of $\mu$.
Then, the restriction
$n^{1/2}(\bar{X} - \mu) \in \mathcal{S}$
can be written as
\begin{align*}
\mu' \widehat{\Sigma}^{-1}\mu 
-2 \mu' \widehat{\Sigma}^{-1}\bar{X}
\leq n^{-1} \kappa(1-\alpha) - \bar{X}' \widehat{\Sigma}^{-1}\bar{X}. 
\end{align*}
This is a convex quadratic constraint in $\mu$.

When some of the constraints $g_1(\psi, \mu) \leq 0$ and $g_2(\psi, \mu) = 0$ are bilinear, 
the resulting feasible region may not be convex. 
To deal with the bilinear constraints, we solve optimization problems using mixed integer programming (MIP) with \texttt{Gurobi} in R.
By virtue of the developments in MIP solvers and fast computing environments, 
MIP has become increasingly used in recent applications. 
For example, 
\cite{bertsimas2016} adopted an MIO approach for obtaining  $\ell_0$-constrained estimators in high-dimensional regression models
and   
\cite{Reguant:2016}  used mixed integer linear programming for computing counterfactual outcomes in game theoretic models.

\section{An Empirical Application}\label{sec:Exmp:new}

\citet{AE1998} use data from the 1980 and 1990 U.S. census to estimate a model of the relation between the number of weeks per year a woman works and the number of children she has. A simplified but nonparametric version of their model is 
\begin{align}\label{FH-iv-AE}
Y = \phi(D) + U; \; \mathbb{E}[U | Z ] = 0,
\end{align}
where $Y$ is the number of weeks a woman works in a year; $\phi$ is an unknown function; and $D = 0,1$, or $2$ according to whether a woman has 2, 3, or 4 or more children.
$D$ is endogenous. $Z$ is a binary instrument for $D$ equal to 1 if the first two children are of the same sex and 0 otherwise. 
We obtain bounds on the partially identified parameter $\phi(0)-\phi(1)$, which measures the change in
the number of weeks a woman works when the number of children she has increases from two to three. 
We use data consisting of 394,840 observations 
$(Y,D,Z)$ from the 1980 U.S. census \citep{IPUMS}.

We assume that $\phi$ is monotone non-increasing and focus on the parameter
$f(\psi) = \phi(0)-\phi(1)$, where
 $\psi = [\phi(0), \phi(1), \phi(2)]'$. 
The population mean vector  is 
$
\mu = [\textrm{vec} (\Pi),  \textrm{vec} (\nu)],
$
where
\begin{align*}
\Pi =
\begin{pmatrix}
p_{DZ}(0,0) & p_{DZ}(1,0) &  p_{DZ}(2,0) \\
p_{DZ}(0,1) & p_{DZ}(1,1) &  p_{DZ}(2,1)
\end{pmatrix},
\ \ \ \
\nu = 
\begin{pmatrix}
\mathbb{E}[ Y I(Z = 0) ] \\
\mathbb{E}[ Y I(Z = 1) ] 
\end{pmatrix},
\end{align*}
$p_{DZ} (d, z) := \text{Pr}( D = d, Z = z)$,
and $I(\cdot)$ is the indicator function.
The dependent variable $Y$ is contained in the interval [0,52]. In the analysis described below, we divide $Y$ by 52, so it is
contained in the interval $[0,1]$, and denote the empirical analog of $\mu$ by $\bar{X}$.
The inequality constraints are
\begin{align}\label{g1-ineq-fh}
\begin{split}
g_1(\psi, \mu)  = 
\begin{pmatrix}
- 1 & 1 & 0 \\
0 & - 1 & 1 
\end{pmatrix}
\begin{pmatrix}
\phi(0) \\
 \phi(1) \\
 \phi(2)
\end{pmatrix}
\leq 0.
\end{split}
\end{align}
The equality constraints are
\begin{align}\label{g1-eq-fh}
\begin{split}
g_2(\psi, \mu) = 
\begin{pmatrix}
\Pi \psi - \mu \\
\sum_{d=0}^2 \sum_{z=0}^1 p_{DZ}(d,z) - 1 
\end{pmatrix}
= 0.
\end{split}
\end{align}
The instrumental variable constraints in 
\eqref{g1-eq-fh} are bilinear in the sense that $\Pi \psi$ contains 6 bilinear terms.

 We estimate the following 8 identification or 95\% confidence intervals for $f(\psi)$.

\begin{enumerate}
\item (\texttt{Sample}) The identification interval given by \eqref{obj-pop}-\eqref{g-constraint}
under the assumption that the sample analogs of the components of $\mu$ equal the population values. 
The resulting estimated identification interval is a consistent point estimate of the population identification interval but
does not take account of random sampling errors in the estimate of $\mu$.

\item (\texttt{Box}) The 95\% confidence interval given by \eqref{obj0}--\eqref{g-constraint-est}  with $\mathcal{S}$ a box, the $D_j$'s equal to the
diagonal elements of the sample covariance matrix $\hat{\Sigma}$, and
$\kappa_b (1-\alpha) = 2.66$, which is obtained via the method given in Section~\ref{sec:algorithms}.
This estimate treats $\hat{\Sigma}$ as if it were the population covariance matrix.

\item (\texttt{Chi-Square})
The 95\% confidence interval given by \eqref{obj0}--\eqref{g-constraint-est} 
 with $\mathcal{S}$ an ellipsoid, $\Upsilon = \hat{\Sigma}$, and
$\kappa_e (1-\alpha) = 14.08$, which is the 0.95 quantile of the chi-square distribution with 7 degrees of
freedom (the number of distinct components of $\mu$ taking account of the constraint that
$\sum_{d=0}^2 \sum_{z=0}^1 p_{DZ} (d,z) = 1$). This estimate, like the previous one, treats $\hat{\Sigma}$ as if it were the population covariance matrix.

\item (\texttt{Method 2 ($\Upsilon = \hat{\Sigma}$, $\sigma^2 = 1$)})
The 95\% confidence interval obtained from Method 2 with $\mathcal{S}$ an ellipsoid, $\Upsilon = \hat{\Sigma}$, $\sigma^2 = 1$, and
$\kappa_{SG,1}(0.95) = 78.89$.
This estimate also treats $\Upsilon = \widehat{\Sigma}$ as if it were non-random.
The choices of $\Upsilon$ and $\sigma^2$ are motivated by 
consistency of $\hat{\Sigma}$ for $\Sigma$.

\item (\texttt{Method 2 ($\Upsilon = I_7$, $\sigma^2 = 1$)})
The 95\% confidence interval obtained from Method 2 with $\mathcal{S}$ an ellipsoid, $\Upsilon = I_7$, $\sigma^2 = 1$, and
$\kappa_{SG,1}(0.95) = 78.89$.
This estimate does not use $\widehat{\Sigma}$.
This choices of $\Upsilon$ and $\sigma^2$ ensure  finite sample validity because each element of $X-\mu$ is contained in the interval $[-1,1]$.

\item (\texttt{Method 2 ($\Upsilon = I_7$, $\sigma^2 = \hat{\sigma}^2 (\log n)$)})
The 95\% confidence interval obtained from Method 2 with $\mathcal{S}$ an ellipsoid, $\Upsilon = I_7$, $\sigma^2 = \hat{\sigma}^2 (\log n) = 0.22$ (that is, the variance proxy is estimated by 
\eqref{est-SG-parameter} with $t = \log n$), and $\kappa_{SG,1}(0.95) = 78.89$.
This estimate also does not use $\widehat{\Sigma}$ but estimates the variance proxy instead of using the bounded support of $X-\mu$.

\item (\texttt{Method 3})
The 95\% confidence interval obtained from Method 3 with $\mathcal{S}$ an ellipsoid, $\Upsilon = \hat{\Sigma}$, $\sigma^2 = 1$, and
$\kappa_{SG,2}(0.95) = 22.15$.
This estimate treats $\Upsilon = \widehat{\Sigma}$ as if it were not random.
The choices of $\Upsilon$ and $\sigma^2$ are motivated by consistency of $\hat{\Sigma}$ for $\Sigma$.

\item (\texttt{Minsker}) The 95\% confidence interval based on 
\cite{Minsker2015}. 
Specifically, it is obtained by replacing $n^{1/2} (\bar{X} - \mu) \in \mathcal{S}$
with $(\hat{\mu}_\ast - \mu)'(\hat{\mu}_\ast - \mu) \leq r_{n,\ast}^2$, where 
$\hat{\mu}_\ast$ is a median-of-means estimator of $\mu$ based on coordinate-wise medians
and $r_{n,\ast}^2$ is the critical value that ensures finite sample coverage. 
See Online Appendix~\ref{sec:alt-theory} for details.

\end{enumerate}

We do not consider a confidence interval based on \eqref{CI-chi2-finite}
because $\textrm{B}(n,p,\overline{\mu}_3) \approx 2.49 \overline{\mu}_3$ (here,  $n = 394,840$ and $p = 7$)  
and as a result, it is extremely unlikely that $\textrm{B}(n,p,\overline{\mu}_3) < 0.05$,
although $\overline{\mu}_3$ is unknown.
Computation was carried out on a MacBookPro laptop with an Apple M1 chip and 16 GB of memory.

\begin{table}[htbp]
\small
\caption{\citet{AE1998} Revisited}\label{ae-example-SG}
\begin{center}
\begin{tabular}{lcc}
\hline\hline
 $\mathcal{S}$ & Bounds & Computing \\ 
               &  &  Time (sec) \\ \hline
   \texttt{Sample} &  $[1.44, 6.36]$ & 0.08  \\
    \texttt{Box} & $[0, 17.06]$ & 11 \\
    \texttt{Chi-Square} & $[0, 10.85]$ & 12  \\  
    \texttt{Method 2 ($\Upsilon = \hat{\Sigma}$, $\sigma^2 = 1$)} & $[0, 17.19]$ & 15  \\
    \texttt{Method 2 ($\Upsilon = I_7$, $\sigma^2 = 1$)} & $[0, 36.39]$ & 5  \\
    \texttt{Method 2 ($\Upsilon = I_7$, $\sigma^2 = \hat{\sigma}^2 (\log n)$)} & $[0, 26.73]$ & 12  \\      
    \texttt{Method 3} & $[0, 12.01]$ & 13  \\   
    \texttt{Minsker} & $[0, 37.12]$ & 6  \\     
    \hline
\end{tabular}
\end{center}
\end{table}%

The results, including computing times, are shown in Table~\ref{ae-example-SG}.
The intervals are in weeks,
not weeks divided by 52.
The estimates obtained by methods 1-8 above are labelled in typewriter font:
\texttt{Sample; 
\hspace*{-2ex} Box; 
\hspace*{-2ex} Chi-Square;}
three types of  \texttt{Method 2};
\texttt{Method 3}; and \texttt{Minsker}; respectively.
 As expected, the \texttt{Sample}
interval is narrower than the \texttt{Box} and \texttt{Chi-Square} intervals, and the \texttt{Chi-Square} interval is narrower than the \texttt{Box} interval. 
The \texttt{Method 2 ($\Upsilon = \hat{\Sigma}$, $\sigma^2 = 1$)} and \texttt{Method 3} intervals are wider than the \texttt{Chi-Square} interval.
 The \texttt{Method 2 ($\Upsilon = \hat{\Sigma}$, $\sigma^2 = 1$)} interval is wider
than the \texttt{Method 3} interval, which is based on a stronger assumption about the distribution of the observed random variables. 
The foregoing methods are motivated by consistency of the estimates of the population parameters they depend on.
They do not take account of random sampling error in the estimates.
The \texttt{Method 2 ($\Upsilon = I_7$, $\sigma^2 = 1$)} and \texttt{Minsker} intervals ensure
 finite-sample coverage probabilities of 0.95 but are wider than the other intervals.
The \texttt{Method 2 ($\Upsilon = I_7$, $\sigma^2 = \hat{\sigma}^2 (\log n)$)} interval 
provides a tighter upper bound of 26.73 
 by estimating the variance proxy instead of relying on the bounded support of $X-\mu$.
Only the \texttt{Sample} estimate yields an informative
lower bound on $\phi(0) - \phi(1)$. Excluding the \texttt{Sample} bounds, which do not take account of random
sampling error in the estimate of $\mu$, the upper bounds indicate that an increase in the number of
children a woman has from 2 to 3 reduces her annual employment
at most
 by approximately 11-37 weeks, depending on the estimation method. All of the computing times are small.

\section{Monte Carlo Experiments}\label{sec:MC:new}

This section presents the results of Monte Carlo experiments that investigate the widths and coverage probabilities of nominal 95\% confidence intervals for $f (\psi)$ that are obtained by
using the methods described in Section \ref{sec:toy:exmp}. For reasons explained in Section \ref{sec:MC:design}, we concentrate on the upper confidence limit $\hat{J}_{+}$ and do not investigate $\hat{J}_{-}$. The experiments are designed to mimic the empirical application of Section \ref{sec:Exmp:new}.

\subsection{Design}\label{sec:MC:design}

Data were generated by simulation from model \eqref{FH-iv-AE} as follows. Using the notation of Section \ref{sec:Exmp:new}, set
\begin{align*}
\psi = [ \phi(0), \phi(1), \phi(2) ] =  \frac{1}{52} (23, 18, 16)
\ \ \text{ and } \ \
\Pi =
\begin{pmatrix}
0.31 & 0.13 &  0.05 \\
0.29 & 0.16 &  0.06
\end{pmatrix}.
\end{align*}
The values of $\phi$  are similar to those used in
 \citet{FH:2015}, who considered larger support of $D$,
 and 
the values of $\Pi$ are the same as the sample probabilities of $\Pi$ in the empirical example of Section \ref{sec:Exmp:new}.
Simulated values of $(D,Z)$ were drawn from $\Pi$. 
The outcome variable $Y$ was generated from 
\begin{align}\label{FH-iv-dgp}
Y = \mathbb{E}[ \phi(D) | Z ] + \frac{V}{52}; \ \ V \sim N(0,1).
\end{align}
The parameter of interest is 
$f(\psi) = 52 [ \phi(0) - \phi(1)] = 5$, which is not point-identified.
We impose that each element of $(\psi, \mu)$ is non-negative
and impose the monotonicity constraint \eqref{g1-ineq-fh}. 
The population value of  $[J_{-}, J_{+}]$ is  $[0.381, 5.438]$.
Thus, the data are not very informative about $J_{-}$.
Consequently, we focus on $\hat{J}_{+}$, the estimate of $J_{+}$, in the experiments 
and do not investigate $J_{-}$. 

We carried out experiments with $\mathcal{S}$ an ellipsoid and $\mathcal{S}$ a box. We report the averages and empirical coverage probabilities of $\hat{J}_{+}$ obtained by using 8 methods described in Section \ref{sec:Exmp:new}. 
The nominal coverage probability was 95\% and 
there were 500 Monte Carlo replications per experiment.
In experiments in which $\mathcal{S}$ is a box, 
$\kappa_b(1-\alpha)$ was computed using $10^6$
 random draws from $N(0,\widehat{\Sigma})$.


\begin{figure}[htbp]
\caption{Distributions of $\hat{J}_{+}$}\label{fig-mc}
\begin{center}
\includegraphics[scale=0.35, angle=270]{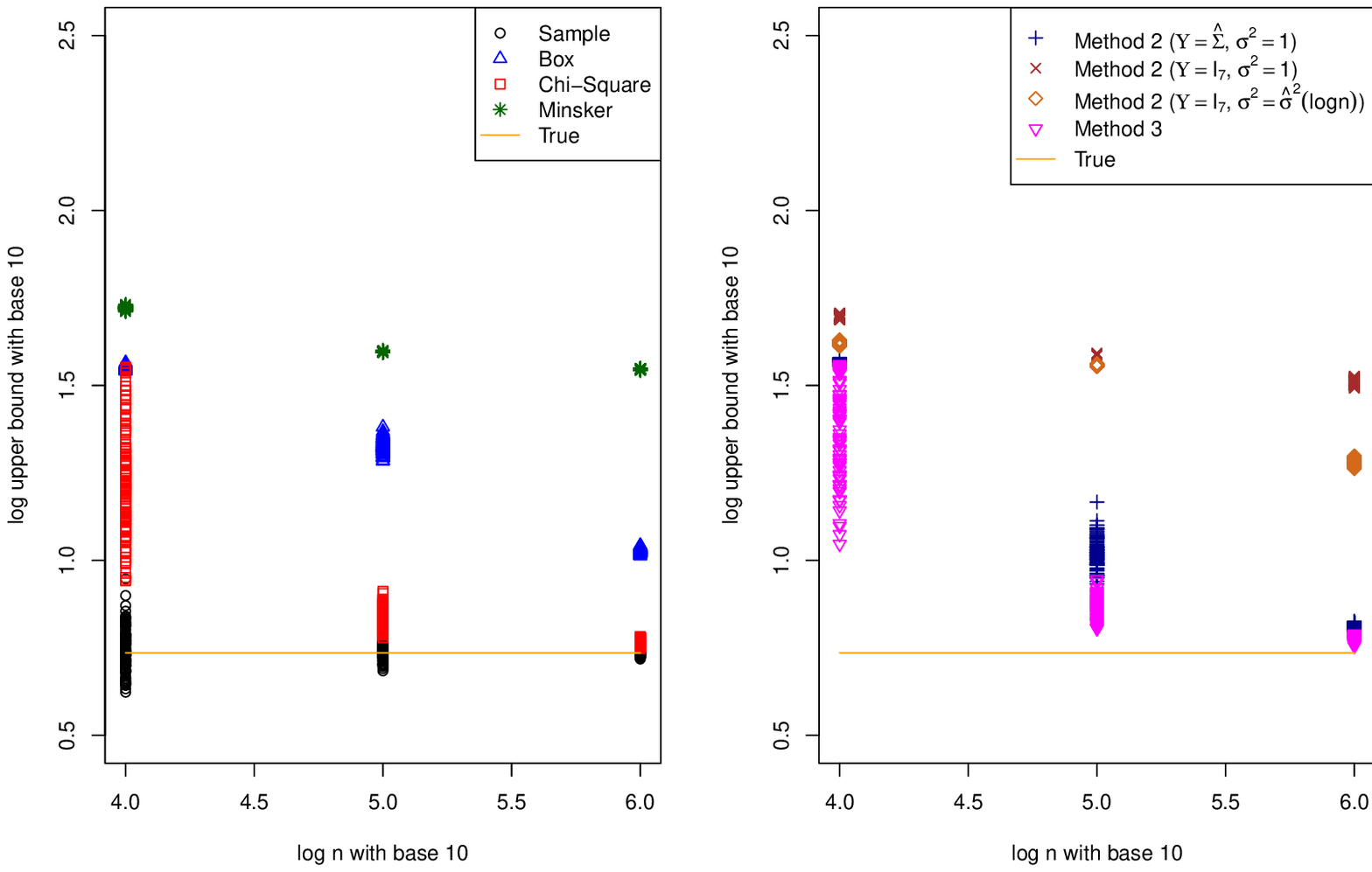}
\end{center}
\end{figure}

\subsection{Results}\label{sec:MC:results}

Figure~\ref{fig-mc} shows only 100 (out of 500) realized values for each different estimated bound by sample size $n
\in \{ 10^4,  10^5,  10^6 \}$.
Each symbol corresponds to one Monte Carlo realization,
and the population value $J_{+}$ is shown as a horizontal line. 
In the figure, both $x$ and $y$ axes are shown in the log scale with base 10.
Table~\ref{MC:tab} summarizes the results of the Monte Carlo experiments. The quantities in parentheses in the table are the values of $\hat{J}_{+}$ for the \texttt{Method 2 ($\Upsilon = \hat{\Sigma}$, $\sigma^2 = 1$)} and \texttt{Method 3} methods obtained with the population values of the sub-Gaussian variance proxies, not estimates.
{In the examples treated in these experiments, inference based on the true variance proxies and on consistent estimates of the proxies would be almost identical.}  

The \texttt{Sample} bounds converge to the true upper bound as $n$ increases; however, they are not suitable for inference since sampling errors are ignored.
All bounds get smaller as $n$ gets large, but 
the  \texttt{Chi-Square}, \texttt{Method 2 ($\Upsilon = \hat{\Sigma}$, $\sigma^2 = 1$)}, and \texttt{Method 3} methods
provide much tighter bounds than the other methods. 
The values of $\hat{J}_{+}$ for the \texttt{Method 2 ($\Upsilon = \hat{\Sigma}$, $\sigma^2 = 1$)} and \texttt{Method 3} methods obtained with estimated and true values of the sub-Gaussian variance proxies are nearly equal.
In every case except the \texttt{Sample} bounds, the estimated bounds are larger than $J_{+}$, resulting in the 100\% empirical coverage probabilities.
Overall, the simulation results are consistent with 
those of the empirical application in Section~\ref{sec:Exmp:new}.


\begin{table}[htbp]
\small
\caption{Results of Monte Carlo Experiments}\label{MC:tab}
\begin{center}
\begin{tabular}{lrrr}
 \hline\hline
 & $n = 10^4$ &  $10^5$ &  $10^6$ \\ 
  \hline
\multicolumn{4}{l}{Averages of $\hat{J}_{+}$} \\  
\texttt{Sample} & 5.60 & 5.43 & 5.45 \\ 
  \texttt{Box} & 35.61 & 21.12 & 10.70 \\ 
  \texttt{Chi-Square} & 17.92 & 6.84 & 5.84 \\ 
  \texttt{Method 2 ($\Upsilon = \hat{\Sigma}$, $\sigma^2 = 1$)} & 36.88 & 10.51 & 6.44 \\ 
\texttt{Method 2 ($\Upsilon = \Sigma$, $\sigma^2 = 1$)} & (36.84) &  &  \\   
  \texttt{Method 2 ($\Upsilon = I_7$, $\sigma^2 = 1$)} & 49.75 & 38.73 & 32.37 \\ 
  \texttt{Method 2 ($\Upsilon = I_7$, $\sigma^2 = \hat{\sigma}^2 (\log n)$)} & 41.78 & 36.09 & 19.07 \\ 
  \texttt{Method 3} & 25.98 & 7.30 & 5.94 \\ 
\texttt{Method 3} with true $\Sigma$ & (25.43) &  &  \\ 
  \texttt{Minsker} & 52.62 & 39.47 & 35.21 \\ 
   \hline
\multicolumn{4}{l}{Proportions that $\hat{J}_{+} \geq J_{+} = 5.438$} \\  
\texttt{Sample} & 0.49 & 0.48 & 0.51 \\ 
  All other methods & 1.00 & 1.00 & 1.00  \\ 
 \hline   
\end{tabular}
\end{center}
\end{table}%

Table~\ref{MC:tab} shows that the bounds on $J_{+}$ vary greatly, depending on the method used
and assumptions made about population parameters. In an application, a researcher might compute the bounds using all the methods and assumptions that are relevant to the application. The researcher can then decide which bounds to use or report, depending on his/her beliefs about the accuracy of the asymptotic approximations s/he is making and how much risk of inaccuracy s/he is willing to accept.

To check sensitivity to sub-Gaussian assumptions, we carried out 
additional Monte Carlo experiments by replacing $V$ in 
\eqref{FH-iv-dgp}
with $V \sim {(\chi^2(5) - 5)}/{\sqrt{10}}$. In other words, we
use a standardized chi-square random variable, which is not sub-Gaussian but sub-exponential. 
The results of the additional Monte Carlo experiments are similar to those reported here with the standard normal $V$.
See Online Appendix~\ref{sec:Exmp:appendix} for details.

\section{Conclusions}\label{sec:concl}

This paper has described a method for carrying out inference on partially identified parameters that are solutions to a class of optimization problems. The parameters arise, for example, in applications in which grouped data are used for estimation of a model's structural parameters. Inference consists of obtaining confidence intervals for the partially identified parameters. The paper has presented three methods for obtaining finite-sample lower bounds on the coverage probabilities of the confidence intervals. The methods correspond to three sets of assumptions of increasing strength about the underlying random variable. With the moderate sample sizes found in most economics applications, the bounds become tighter as the assumptions strengthen. The paper has also described a computational algorithm for implementing the methods. The results of Monte Carlo experiments and an empirical example illustrate the methods' usefulness.
The paper has focused on the case in which a vector of first moments is the only unknown population parameter in the population version of the optimization problem. It might be useful to extend our formulation to the case in which the population optimization problem contains other population parameters. It is an open question how to obtain finite-sample lower bounds on coverage probabilities in such a case.

\appendix

\renewcommand{\thepage}{A-\arabic{page}} \setcounter{page}{1}


\section{Proofs of Theorems}\label{app:proof}

\begin{proof}[Proof of Theorem \ref{thm2}]
Suppose $n^{1/2} (\bar{X} - \mu)$ is in   $\mathcal{S}$.
 Any feasible solution of \eqref{obj-pop}--\eqref{g-constraint} 
 is also a feasible solution of 
 \eqref{obj0}--\eqref{g-constraint-est}.  
 Therefore, the feasible region of 
 \eqref{obj0}--\eqref{g-constraint-est}  
 contains the feasible region of 
 \eqref{obj-pop}--\eqref{g-constraint}.  
 Consequently,
 \[
\hat{J}_{-}(\bar{X}) \leq J_{-} \leq J_{+} \leq \hat{J}_{+}(\bar{X}),
\]
which in turn proves \eqref{main-result}.
\end{proof}

\begin{proof}[Proof of Theorem \ref{lem1}]
Define the random $p$-vector $\bar{V} := \Sigma^{-1/2} \bar{Z}$. Then
$\mathbb{E}( \bar{V} ) = 0$ and $\text{cov}(\bar{V}) = I_{p \times p}$. Define the set
\begin{align}\label{S-set-def}
\mathcal{S}_\Sigma := \left\{ \xi : \xi = \Sigma^{-1/2} \zeta; \zeta \in \mathcal{S} \right\}.
\end{align}
Then $\mathcal{S}_\Sigma$ is a convex set. Define the random
vectors $U_i \sim N(0, I_{p \times p})$ and
$\bar{U} := n^{1/2} \sum_{i=1}^n U_i$.
It follows from Assumption \ref{a2}(ii) and the generalized Minkowski
inequality that
\[
\mathbb{E} \left[ \left( \sum_{j=1}^p [ \Sigma^{-1/2} (X_i - \mu) ]_j^2 \right)^{3/2} \right] \leq p^{3/2} \overline{\mu}_3.
\]
In addition, it follows from Theorem 1.1 of \cite{Raic2019} that
\[
\abs{ \mathbb{P} (\bar{V} \in \mathcal{S}_\Sigma) - \mathbb{P} (\bar{U} \in \mathcal{S}_\Sigma) }
\leq \frac{(42p^{1/4} + 16) p^{3/2} \overline{\mu}_3}{ n^{1/2}},
\]
which proves the theorem.
\end{proof}

\begin{proof}[Proof of Theorem~\ref{lem-SG-1}]
Let $\norm{a}$ denote the Euclidean norm for a vector $a$
and $\tilde{Z}_i := [ \Upsilon^{-1/2} (X_i - \mu) ]$.
Given any $r > 0$, we have that 
\begin{align*}
\mathbb{P} \left( \norm{ n^{-1} \sum_{i=1}^n \tilde{Z}_i} > r \right)
&=_{(1)} 
\mathbb{P} \left(  \sum_{j=1}^p \left\{ n^{-1} \sum_{i=1}^n \tilde{Z}_{ij} \right\}^2  > r^2 \right) \\
&\leq_{(2)}  
\mathbb{P} \left(  \bigcup_{j=1}^p  \left\{ n^{-1} \sum_{i=1}^n \tilde{Z}_{ij} \right\}^2  > \frac{r^2}{p} \right) \\
&\leq_{(3)}  
\sum_{j=1}^p \mathbb{P} \left(   \left\{ n^{-1} \sum_{i=1}^n \tilde{Z}_{ij} \right\}^2  > \frac{r^2}{p} \right) \\
&=_{(4)}  
\sum_{j=1}^p \mathbb{P} \left(  \left| n^{-1} \sum_{i=1}^n \tilde{Z}_{ij} \right|  > \frac{r}{\sqrt{p}} \right) \\
&\leq_{(5)}  
2\sum_{j=1}^p 
\exp \left( - \frac{n r^2}{2 p \sigma_j^2} \right),
\end{align*}
where 
(1) is obtained by squaring both sides of the inequality,
(2) comes from the fact that if $a_1 + \ldots + a_p > c$, then $a_j > c/p$ for at least one $j$,
(3) follows because the probability of a union of events is bounded by the sum of probabilities,
(4) is by taking the square root of the both sides of the inequality,
and 
(5) is by applying Hoeffding bound for the sum of independent sub-Gaussian random variables
\citep[see, e.g., Proposition 2.5 of ][]{wainwright2019book}.
In particular, the result above implies that if we take $\sigma^2 = \max_{1 \leq j \leq p} \sigma_j^2$, 
\begin{align*}
\mathbb{P} \left\{ \frac{n}{\sigma^2} (\bar{X} - \mu)' \Upsilon^{-1} (\bar{X} - \mu) > \kappa_{SG, 1}(1-\alpha)  \right\}
\leq  
\alpha, 
\end{align*}
which proves the theorem.
\end{proof}

\begin{proof}[Proof of Theorem~\ref{lem-SG-2}]
It follows from Theorem 2.1 of \citet{HKT:2021} (by taking $A$ in their theorem to be the $p$-dimensional identity matrix) that 
\begin{align*}
\mathbb{P} \left\{ \frac{n}{\sigma^2} (\bar{X} - \mu)' \Upsilon^{-1} (\bar{X} - \mu) > 
p + 2 \sqrt{ p t } + 2 t 
  \right\}
\leq  
\exp(-t)
\end{align*}
for all $t > 0$.
This yields the theorem immediately.
\end{proof}

\begin{proof}[Proof of Lemma \ref{lem2}]
Define 
$$
\Delta_n := \sup_{\mathcal{S}} \abs{   \mathbf{P} (\mathcal{S}, \widehat{\Sigma}) - \mathbb{P} (W \in \mathcal{S}) }.
$$
Now
\[
\mathbf{P} (\mathcal{S}, \widehat{\Sigma}) - \mathbb{P} (W \in \mathcal{S})
= \mathbb{P} (\Sigma^{-1/2} \widehat{W} \in \mathcal{S}_\Sigma) - \mathbb{P} (\xi \in \mathcal{S}_\Sigma),
\]
where $\xi \sim N(0, I_{p \times p})$ and $\mathcal{S}_\Sigma$ is defined in \eqref{S-set-def}.
Therefore,
\begin{align*}
\Delta_n &= \sup_{\mathcal{S}} \abs{   \mathbb{P} (\Sigma^{-1/2} \widehat{W} \in \mathcal{S}_\Sigma) 
- \mathbb{P} (\xi \in \mathcal{S}_\Sigma) } \\
&\leq \text{TV} \left[ N(0, I_{p \times p}), N(0, \Sigma^{-1} \widehat{\Sigma}) \right],
\end{align*}
where $\text{TV}(P_1, P_2)$ is the total variation distance between distributions $P_1$ and $P_2$. By
Example 2.3 of \cite{Dasgupta08},
\[
\text{TV} \left[ N(0, I_{p \times p}), N(0, \Sigma^{-1} \widehat{\Sigma}) \right] \leq p 2^{p+1} \norm{ \Sigma^{-1} \widehat{\Sigma} -  I_{p \times p}}_{\mathrm{F}}
\]
where for any matrix $A$, 
$$
\norm{A}^2_{\mathrm{F}} := \sum_{j=1}^p \sum_{j=1}^p a_{jk}^2.
$$
Define $\omega := \widehat{\Sigma} - \Sigma$. 
Then,
\[
\Sigma^{-1} \widehat{\Sigma} - I_{p \times p} = \Sigma^{-1} (\widehat{\Sigma} - \Sigma) = \Sigma^{-1} \omega,
\]
\[
\abs{( \Sigma^{-1} \omega)_{jk} } 
\leq \sum_{\ell=1}^p \abs{ \Sigma_{j\ell}^{-1} \omega_{\ell k}}
\leq C_{\Sigma} \sum_{\ell=1}^p \abs{  \omega_{\ell k}}
\]
and
\[
\norm{ \Sigma^{-1} \widehat{\Sigma} -  I_{p \times p} }_{\mathrm{F}}
\leq C_{\Sigma} p^{1/2} \left[ \sum_{k=1}^p \left( \sum_{\ell=1}^p \abs{ \omega_{\ell k} } \right)^2 \right]^{1/2}.
\]
To obtain the conclusion of the theorem, it remains to show that $\abs{ \omega_{jk} } \leq r_n(t)$
with probability at least $1 - 2 e^{-t}$. 
We prove this claim below.

Write 
\[
\omega = n^{-1} \sum_{i=1}^n \left[ (X_i - \mu)(X_i - \mu)'  - \Sigma \right] -   (\bar{X} - \mu)(\bar{X} - \mu)'.
\]
Then,  we have that
\begin{align*}
&\mathbb{P} \left[ \max_{j,k} \abs{ \omega_{jk} } \geq r_n(t)  \right] \\
&\leq
\mathbb{P} \left[ \max_{1 \leq j,k \leq p} 
\abs{ n^{-1} \sum_{i=1}^n \left[ (X_{ij}-\mu)(X_{ik}-\mu)' - \Sigma_{jk} \right] }
 \geq \frac{r_n(t)}{2} \right]
 +
\mathbb{P} \left[ \left\{ \max_{1 \leq j \leq p} 
\abs{ \bar{X}_j - \mu_j } \right\}^2
 \geq \frac{r_n(t)}{2} \right]. 
 \end{align*}
Recall that 
\begin{align*}
\Sigma_{jj} = \mathbb{E} \left[ \abs{ X_{ij} - \mu_j }^2 \right]
\; \text{ and } \;
\Gamma_{jk} = \mathbb{E} \left[ \abs{ (X_{ij} - \mu_j)(X_{ik} - \mu_k) - \Sigma_{jk} }^2 \right]. 
\end{align*}
Define
\begin{align*}
\lambda (\kappa_1, n, p) :=
\sqrt{\frac{2 \log (2p) }{n}} + \frac{\kappa_1 \log (2p) }{n}.
\end{align*}
By 
Bernstein's inequality for the maximum of $p$ averages 
(see, e.g., \citet[Lemma 14.13]{buhlmann2011statistics}),
\[
\mathbb{P} \left[ \max_{1 \leq j \leq p} 
\abs{ \Sigma_{jj}^{-1/2} (\bar{X}_j - \mu_j) } \geq  \frac{\kappa_1 t}{\Sigma_{jj}^{1/2} n} + \sqrt{\frac{2t}{n} } 
+ \lambda \left( \Sigma_{jj}^{-1/2} \kappa_1, n, p \right) \right]  \leq 
 \exp \left( -t \right)
\]
and
\[
\mathbb{P} \left[ \max_{1 \leq j,k \leq p} 
\abs{ \Gamma_{jk}^{-1/2}  n^{-1} \sum_{i=1}^n \left[ (X_{ij}-\mu)(X_{ik}-\mu)' - \Sigma_{jk} \right] } \geq \frac{\kappa_1 t}{\Gamma_{jk}^{1/2} n} + \sqrt{\frac{2t}{n} } 
+ \lambda (\Gamma_{jk}^{-1/2} \kappa_1, n, p^2) \right]  \leq 
 \exp \left( -t \right).
\]
Suppose that for each $j$,
\begin{align}\label{tech:condition:app}
\max\{ t, \log(2p) \} \leq \min  \left\{ \frac{n}{4 \kappa_1}, \frac{n}{32 \Sigma_{jj}} \right\}.
\end{align}
Under \eqref{tech:condition:app},
\begin{align*}
&\mathbb{P} \left[ \left\{ \max_{1 \leq j \leq p} 
\abs{ \bar{X}_j - \mu_j } \right\}^2 \geq  \frac{\kappa_1 t}{n} + \Sigma_{jj}^{1/2} \sqrt{\frac{2t}{n} } 
+ \Sigma_{jj}^{1/2} \lambda \left( \Sigma_{jj}^{-1/2} \kappa_1, n, p \right) \right]  \\
&\leq 
 \mathbb{P} \left[ \left\{ \max_{1 \leq j \leq p} 
\abs{ \bar{X}_j - \mu_j } \right\}^2
 \geq 
 \left\{
 \frac{\kappa_1 t}{n} + \Sigma_{jj}^{1/2} \sqrt{\frac{2t}{n} } 
+ \Sigma_{jj}^{1/2} \lambda \left( \Sigma_{jj}^{-1/2} \kappa_1, n, p \right) \right\}^2 \right] \\
& \leq
 \exp \left( -t \right).
\end{align*}
Therefore, if we can choose $r_n(t)$ by 
\begin{align}\label{choice:rt}
\frac{r_n(t)}{2} & :=
\max_{j,k} \left\{ \Sigma_{jj}^{1/2}, \Gamma_{jk}^{1/2} \right\} 
\left\{ \sqrt{\frac{2t}{n} } 
+ 
 \sqrt{\frac{2 \log (2p) }{n}} \right\}
+ \frac{\kappa_1 t}{n}  
+ \frac{\kappa_1 \log (2p) }{n} 
\end{align}
for  $t > 0$, 
we have that 
\[
\mathbb{P} \left[ \max_{j,k} \abs{ \omega_{jk} } \leq r_n(t) \right]
\geq 1 - 2 e^{-t}.
\]
To simplify the upper bound $r_n(t)$, we now strengthen \eqref{tech:condition:app} to \eqref{tech:condition:further} stated in the theorem.
Then, we can take 
\begin{align*}
r_n(t) & :=
  8 \sqrt{\frac{2 \kappa_0 t}{n} }, 
\end{align*}
which proves the claim. 
\end{proof}

\begin{proof}[Proof of Theorem \ref{thm1}]
Write
\begin{align*}
\abs{ \mathbb{P} (\bar{Z} \in \mathcal{S}) - \mathbf{P} (\mathcal{S}, \widehat{\Sigma}) }
&= \abs{ \left[ \mathbb{P} (\bar{Z} \in \mathcal{S}) - \mathbb{P} (W \in \mathcal{S}) \right] -  \left[ \mathbf{P} (\mathcal{S}, \widehat{\Sigma}) - \mathbb{P} (W \in \mathcal{S}) \right]}  \\
&\leq \abs{  \mathbb{P} (\bar{Z} \in \mathcal{S}) - \mathbb{P} (W \in \mathcal{S})  } 
+ \abs{   \mathbf{P} (\mathcal{S}, \widehat{\Sigma}) - \mathbb{P} (W \in \mathcal{S}) }.
\end{align*}
Thus, the theorem follows immediately by combining Theorem \ref{lem1} and Lemma \ref{lem2}. 
\end{proof}

\begin{proof}[Proof of Theorem \ref{thm3}]
Combining Theorems \ref{thm2} and \ref{thm1} yields Theorem \ref{thm3}.
\end{proof}

\begin{proof}[Proof of Lemma~\ref{lem:sample-var}]
By Lemma 1.12 of \cite{rigollet2015high},
$( \tilde{X}_{ij}^2 - \mathbb{E} [ \tilde{X}_{ij}^2 ] )$
is sub-exponential with parameter $16 \sigma_j^2$. 
This in turn implies that  by Theorem 1.13 of \cite{rigollet2015high}, which is a version of 
Bernstein’s inequality, we have, for any $s > 0$,  
\begin{align*}
 \mathbb{P} \left[ 
\left| n^{-1} \sum_{i=1}^n  \tilde{X}_{ij}^2 - \mathbb{E} [ \tilde{X}_{ij}^2 ] \right| > s
 \right] 
 &\leq
 \exp \left[ 
 - \frac{n}{2} \min 
 \left\{ \frac{s^2}{ (16 \sigma_j^2)^2 },  \frac{s}{ 16 \sigma_j^2 } \right\} 
 \right].
\end{align*}
Take $s =  (t / n)^{1/2}$. 
Then, for any $0 < t \leq (16 \sigma_j^2)^2 n$, 
\begin{align}
 \mathbb{P} \left[ 
\left| n^{-1} \sum_{i=1}^n  \tilde{X}_{ij}^2 - \mathbb{E} [ \tilde{X}_{ij}^2 ] \right| >  \left( \frac{t}{n} \right)^{1/2}
 \right] 
\leq e^{- t / (512 \sigma_j^4) }.  \label{lem-res1}
\end{align}
As before, by Hoeffding bound for the sum of independent sub-Gaussian random variables
\citep[see, e.g., Proposition 2.5 of ][]{wainwright2019book}, we have that 
\begin{align*}
\mathbb{P} \left[ 
 \left| n^{-1} \sum_{i=1}^n \tilde{X}_{ij} - \mathbb{E} [ \tilde{X}_{ij} ] \right|
\geq \left( \frac{ t}{n} \right)^{1/2} 
\right] 
\leq
2 e^{-t / (2 \sigma_j^2 )}.
\end{align*}
Furthermore, using the fact that
\begin{align*}
\left| \left( n^{-1} \sum_{i=1}^n \tilde{X}_{ij} \right)^2 - ( \mathbb{E} [ \tilde{X}_{ij} ] )^2 \right|
&\leq
\left( n^{-1} \sum_{i=1}^n \tilde{X}_{ij} - \mathbb{E} [ \tilde{X}_{ij} ] \right)^2
+ 2 \left| \mathbb{E} [ \tilde{X}_{ij} ] \right| \left| n^{-1} \sum_{i=1}^n \tilde{X}_{ij} - \mathbb{E} [ \tilde{X}_{ij} ] \right| \\
&\leq
3 \left( n^{-1} \sum_{i=1}^n \tilde{X}_{ij} - \mathbb{E} [ \tilde{X}_{ij} ] \right)^2
+ 2 \left| n^{-1} \sum_{i=1}^n \tilde{X}_{ij} \right| \left| n^{-1} \sum_{i=1}^n \tilde{X}_{ij} - \mathbb{E} [ \tilde{X}_{ij} ] \right|,
\end{align*}
we have that 
\begin{align}\label{lem-res2}
\begin{split}
& \mathbb{P} \left[ 
\left| \left( n^{-1} \sum_{i=1}^n \tilde{X}_{ij} \right)^2 - ( \mathbb{E} [ \tilde{X}_{ij} ] )^2 \right|
\geq
3 \left( \frac{ t}{n} \right)  + 
2 \left| n^{-1} \sum_{i=1}^n \tilde{X}_{ij} \right| \left( \frac{ t}{n} \right)^{1/2} 
\right] \\
&\leq
\mathbb{P} \left[ 
 \left| n^{-1} \sum_{i=1}^n \tilde{X}_{ij} - \mathbb{E} [ \tilde{X}_{ij} ] \right|
\geq \left( \frac{ t}{n} \right)^{1/2} 
\right] 
\leq
2 e^{-t / (2 \sigma_j^2)}.  
\end{split}
\end{align}
The desired result follows by combining \eqref{lem-res2} with \eqref{lem-res1}.
\end{proof}

\begin{proof}[Proof of Lemma~\ref{lem:sample-var-proxy-equ}]
By Assumption~\ref{a4}, $\tilde{Z}_{ij} = [ \Upsilon^{-1/2} (X_i - \mu) ]_j$
is a mean-zero scalar sub-Gaussian random variable with variance proxy $\sigma^2$, which is denoted by
$\tilde{Z}_{ij} \sim  \textrm{subG}(\sigma^2)$.
Let 
$\bar{\tilde{Z}}_j := n^{-1/2} \sum_{i=1}^n \tilde{Z}_{ij}$. Then $\bar{\tilde{Z}}_j \rightarrow_d N(0,\mathbb{E} [ \tilde{Z}_{ij}^2 ])$ and $\bar{\tilde{Z}}_j \sim  \textrm{subG}(\sigma^2)$ for all $n$.
Let $\varphi_{\tilde Z}$ and  $\varphi_n$, respectively, denote the moment-generating functions of 
$\tilde{Z}_{ij}$ and $\bar{\tilde{Z}}_j$.
Then
\begin{align}
\varphi_n(t) &\leq \exp (\sigma^2 t^2 / 2), \label{subG-eq-no1-z}
\end{align}
and
\begin{align*}
\varphi_n(t) &= \mathbb{E} \exp(  \bar{\tilde{Z}}_j t )
= \mathbb{E} \exp \left( \sum_{i=1}^n \frac{t \tilde{Z}_{ij}}{n^{1/2} } \right)
= \prod_{i=1}^n \mathbb{E} \exp \left( \frac{t \tilde{Z}_{ij}}{n^{1/2} } \right)
= \left[ \varphi_{\tilde{Z}} \left( \frac{t}{n^{1/2}} \right) \right]^n.
\end{align*}
Arguments identical to those used to prove the Lindeberg-L\'{e}vy central limit theorem show that
\begin{align}\label{subG-eq-no2-z}
\lim_{n \rightarrow \infty} \varphi_n(t) &= \exp \left( \mathbb{E} [ \tilde{Z}_{ij}^2 ] t^2 / 2 \right).
\end{align}
It follows from \eqref{subG-eq-no1-z} that $\mathbb{E} [ \tilde{Z}_{ij}^2 ] \leq \sigma^2$.
\end{proof}

\begin{proof}[Proof of Theorem \ref{thm3-SG}]
Note that 
\begin{align*}
&
\mathbb{P} \left\{ n (\bar{X} - \mu)' \Upsilon^{-1} (\bar{X} - \mu) > 
\hat{\sigma}^2 (t)
\kappa_{SG, 1}(1-\alpha)  \right\} \\
&=
\mathbb{P} \left\{ n (\bar{X} - \mu)' \Upsilon^{-1} (\bar{X} - \mu) > \hat{\sigma}^2 (t) \kappa_{SG, 1}(1-\alpha)  
\ \ \text{ and } \ \
\hat{\sigma}^2 (t) < \sigma^2
\right\} \\
&\;\;\;+
\mathbb{P} \left\{ n (\bar{X} - \mu)' \Upsilon^{-1} (\bar{X} - \mu) > \hat{\sigma}^2 (t) \kappa_{SG, 1}(1-\alpha)  
\ \ \text{ and } \ \
\hat{\sigma}^2 (t) \geq \sigma^2
\right\} \\
&\leq 
\mathbb{P} \left\{ \hat{\sigma}^2 (t) < \sigma^2
\right\} +
\mathbb{P} \left\{ n (\bar{X} - \mu)' \Upsilon^{-1} (\bar{X} - \mu) > \sigma^2 \kappa_{SG, 1}(1-\alpha)  
\right\}.
\end{align*}
Then, the theorem follows from Lemma~\ref{lem:sample-var} and Theorem~\ref{thm2}.
\end{proof}

\section{Continuous Covariates}\label{subsec:cont}

In this section, we consider the case in which $g_1$ and $g_2$ depend on a continuous covariate $\nu$ in addition to $(\psi, \mu)$. 
This situation occurs, for example, in applications where some
  observed variables are group averages and others are continuously distributed characteristics of individuals. 
 {If $\nu$ is discrete, the results of Section \ref{subsec:analysis} apply after replacing problem \eqref{obj0}-\eqref{g-constraint-est} with \eqref{obj2}-\eqref{g-constraint-est-cont-dis} below.}
When there is a continuous covariate, $\nu$, \eqref{obj0}-\eqref{g-constraint-est}
become 
\begin{align}\label{obj0-cont}
\hat{J}_{+}(\bar{X}) := \max_{\psi, m} f(\psi)
 \; \text{ and } \; 
\hat{J}_{-}(\bar{X}) := \min_{\psi, m} f(\psi)
\end{align}
subject to
\begin{subequations}\label{g-constraint-est-cont}
\begin{align}
g_1(\psi, m, \nu) &\leq 0 \text{ for every $\nu$}, \label{g1-constraint-est-cont} \\
g_2(\psi, m, \nu) &= 0 \text{ for every $\nu$},  \label{g2-constraint-est-cont} \\
\psi &\in \Psi, \\
n^{1/2} (\bar{X} - m) &\in \mathcal{S}. \label{S-constraint-est-cont}
\end{align}
\end{subequations}
Thus, there is a continuum of constraints. We form a discrete approximation to 
\eqref{g1-constraint-est-cont}-\eqref{g2-constraint-est-cont} by restricting $\nu$ to a discrete grid of points. Let $L$ denote the number of grid points. We give conditions under which the optimal values of the objective functions of the discretized version of
\eqref{obj0-cont}-\eqref{g-constraint-est-cont} converge to 
$\hat{J}_{+}(\bar{X})$ and $\hat{J}_{-}(\bar{X})$
as $L \rightarrow \infty$. 
To minimize the notational
complexity of the following discussion we assume that $\nu$ is a scalar. The generalization to a vector is straightforward. We also assume that $\nu$ is contained in a compact set which, without further loss of generality, we take to be $[0,1]$.

To obtain the grid approximation, let $0 = v_0 < v_1 < v_2 < \ldots < v_L = 1$ be a grid of equally space points
in $[0,1]$. The distance between grid points is $1/(L-1)$. Approximate problem \eqref{obj0-cont}-\eqref{g-constraint-est-cont} by
\begin{align}\label{obj2}
\tilde{J}_{+}(\bar{X}) := \max_{\psi, m} f(\psi)
 \; \text{ and } \; 
\tilde{J}_{-}(\bar{X}) := \min_{\psi, m} f(\psi)
\end{align}
subject to the constraints: 
\begin{subequations}\label{g-constraint-est-cont-dis}
\begin{align}
g_1(\psi, m, \nu_\ell) &\leq 0; \; \ell=1,\ldots,L,   \label{g1-constraint-est-cont-dis} \\
g_2(\psi, m, \nu_\ell) &= 0; \; \ell=1,\ldots,L,  \label{g2-constraint-est-cont-dis} \\
\psi &\in \Psi, 
\intertext{and} 
n^{1/2} (\bar{X} - m) &\in \mathcal{S}. \label{S-constraint-est-cont-dis}
\end{align}
\end{subequations}
We then have

\begin{theorem}\label{cont-thm}
Assume that $f$ is continuous, $\nu \in [0,1]$,
 and $m$ in \eqref{g-constraint-est-cont-dis} is contained in a compact set $\mathcal{M}$. Moreover,
\begin{align*}
\abs{g_j(\psi, m; x) - g_j(\psi, m; v_\ell)   } &\leq C \abs{ x - v_\ell } \\
\abs{g_j(\psi, m; x) - g_j(\psi, m; v_{\ell+1})   } &\leq C \abs{ x - v_{\ell+1} }
\end{align*}
for $j =1$ or $2$, some $C < \infty$,  and all $\psi \in \Psi$, all $m \in \mathcal{M}$, all $x \in [v_\ell, v_{\ell+1}] \in [0,1]$. Then
\begin{align}\label{cont-result}
\lim_{L \rightarrow \infty} \tilde{J}_{+} = \hat{J}_{+}
\; \text{ and } \;
\lim_{L \rightarrow \infty} \tilde{J}_{-} = \hat{J}_{-}.
\end{align}
\end{theorem}

Theorem \ref{cont-thm} implies that under weak smoothness assumptions, a sufficiently dense grid provides an arbitrarily accurate approximation to the continuously constrained optimization problem \eqref{obj0-cont}-\eqref{g-constraint-est-cont}.

\begin{proof}[Proof of Theorem \ref{cont-thm}] 
We focus on the maximization problem since the minimization problem can be analyzed analogously. 

Let $(\psi_L, \mu_L)$ denote the optimal solution to the maximization version of \eqref{obj2}-\eqref{g-constraint-est-cont-dis}.
Define $g(\psi, \mu; \nu) = [ g_1(\psi, \mu; \nu), g_2(\psi, \mu; \nu), -g_2(\psi, \mu; \nu) ]$, so that $g(\psi, \mu; \nu) \leq 0$ componentwise.
Define $\ell(\nu) := \text{arg min}_\ell \abs{\nu - \nu_\ell}$.
Then
\[
\sup_{\nu \in [0,1]} \abs{ g(\psi, \mu; \nu) - g(\psi, \mu; \ell(\nu)) } \leq C/(L-1)
\]
componentwise and $g(\psi, \mu; \nu_\ell) \leq 0$ implies that
\[
g(\psi, \mu; \nu) \leq C/(L-1)
\]
componentwise uniformly over $\nu \in [0,1]$. 
Therefore, $(\psi_L, \mu_L)$ is a feasible solution to
\begin{align}\label{obj-c1}
J_{+}^\ast := \max_{\psi, \mu} f(\psi)
\end{align}
subject to the new constraint:
\begin{align*}
g(\psi, \mu; \nu) \leq C/(L-1)  
 \; \text{for all $\nu \in [0,1]$, $\nu = \mathrm{rational}$, and} \; \
n^{1/2} (\bar{X} - \mu) \in \mathcal{S}.
\end{align*}
Consequently, $J_{+}^\ast \geq \tilde{J}_{+} \geq J_{+}$,
where $J_{+} = \max f(\psi)$ subject to $g_1 (\psi,\mu,x) \leq 0$, $g_2(\psi,\mu,x)=0$ , and $\psi \in \Psi$.
 Define
\begin{align*}
\Pi := \Big\{ \xi \geq 1, \eta: &\text{ there is $(\psi, \mu)$ such that $n^{1/2} (\bar{X} - \mu) \in \mathcal{S}$,
$f(\psi) \leq \eta$,} \\ 
&\text{ and $g(\varphi, \mu; \nu) \leq C/\xi$ for all $\nu \in [0,1]$} \Big\}.
\end{align*}
Note that $\Pi$ is a closed set. Therefore, by Proposition 3.3 of \cite{JW90}, 
$J^\ast_{+} \rightarrow J_{+}$ 
as $L \rightarrow \infty$ if the constraints are restricted to rational
values of $\nu \in [0,1]$. It follows from continuity of $g$ as a function of $\nu$ that the
constraints hold for all $\nu \in [0,1]$.
\end{proof}

\section{Minsker's (2015) Median of Means Method}\label{sec:alt-theory}

In this  appendix, we describe how to carry out inference based on  \cite{Minsker2015}. 
In particular, we consider two versions of the median of means: the one based on geometric median and 
the other using coordinate-wise medians. 
 {\cite{lugosi2019} propose a different version of the median of means estimator that has theoretically better properties but is more difficult to compute.} 

First, for the case of geometric median, let $\alpha_\ast := 7/18$ and $p_\ast := 0.1$. 
Define
\begin{align*}
\psi (\alpha_\ast; p_\ast) = (1-\alpha_\ast) \log \frac{1-\alpha_\ast}{1-p_\ast} + \alpha_\ast \log \frac{\alpha_\ast}{p_\ast}.
\end{align*}
Let $0 < \delta < 1$ be the level of the confidence set and set
\begin{align}\label{def-k}
k := \left\lfloor  \frac{\log (1/\delta)}{\psi (\alpha_\ast; p_\ast)} \right\rfloor + 1.
\end{align}
Assume that $\delta$ is small enough that $k \leq n/2$. 
Divide the sample $X_1,\ldots,X_n$ into $k$ disjoint groups $G_1,\ldots,G_k$ of size $\left\lfloor  \frac{n}{k} \right\rfloor$ each, and define
\begin{align*}
\hat{\mu}_j &:= \frac{1}{| G_j |} \sum_{i \in G_j} X_i, \; j=1,\ldots,k, \\
\hat{\mu} &:= \text{G.med}(\hat{\mu}_1,\ldots,\hat{\mu}_k),
\end{align*}
where \text{G.med} refers to the geometric median. 
See \cite{Minsker2015} and references therein for details on the geometric median.
The intuition behind $\hat{\mu}$ is that it is a robust measure of the population mean vector $\mu$ since 
each subsample mean vector $\hat{\mu}_j$ is an unbiased estimator for $\mu$ and the aggregation method via  the geometric median is robust to outliners.
Because of this feature, it turns out that the finite sample bound for the Euclidean norm distance between $\hat{\mu}$ and $\mu$ depends only on  $\text{tr}(\Sigma)$, but not on the higher moments \citep[see Corollary 4.1 of][]{Minsker2015}. This is the main selling point of the median of means since the finite sample probability bound for the usual sample mean assumes the existence of a higher moment (e.g. the third absolute moment in \citet{Bentkus03} and Theorem \ref{lem1} in Section \ref{subsec:analysis}). 

Second, \cite{Minsker2015} also considered using coordinate-wise medians instead of using the geometric median.
In this case, let $\alpha_\ast = 1/2$ and $p_\ast = 0.12$. Then $k$ is redefined via \eqref{def-k}. Let $\hat{\mu}_\ast$ denote the vector of coordinate-wise medians.

To estimate $\text{tr}(\Sigma)$, \cite{Minsker2015} proposed the following:
\begin{align*}
\hat{T}_j &:= \frac{1}{| G_j |} \sum_{i \in G_j} \norm{ X_i - \hat{\mu}_j}^2, \; j=1,\ldots,k, \\
\hat{T} &:= \text{med}(\hat{T}_1,\ldots,\hat{T}_k).
\end{align*}
where $\norm{ a }$  is the Euclidean norm of a vector $a$.
Let $B(h,r)$ denote the ball of radius $r$ centered at $h$
and let 
\begin{align*}
r_n &:= 11 \sqrt{2} \sqrt{\hat{T} \frac{\log(1.4/\delta)}{n} }, \\
r_{n,\ast} &:= 4.4 \sqrt{2} \sqrt{\hat{T} \frac{\log(1.6 d_\mu /\delta)}{n - 2.4\log(1.6 d_\mu /\delta)} },
\end{align*} 
where $d_\mu$ is the dimension of $\mu$.

\begin{lemma}[\cite{Minsker2015}]\label{minsker-thm}
Assume that 
\begin{align}\label{Minsker-ineq}
15.2 \sqrt{\frac{\mathbb{E} \norm{ X - \mu}^4 - (\text{tr}(\Sigma))^2}{(\text{tr}(\Sigma))^2} }
\leq
\left( \frac{1}{2} - 178 \frac{\log(1.4/\delta)}{n} \right) \sqrt{\frac{n}{\log(1.4/\delta)} }.
\end{align}
Then
\begin{align}\label{Minsker-result}
\mathbb{P} \left[ \mu \in B(\hat{\mu}, r_n) \right] \geq 1 - 2\delta
\ \ \text{ and } \ \
\mathbb{P} \left[ \mu \in B(\hat{\mu}_\ast, r_{n,\ast}) \right] \geq 1 - 2\delta.
\end{align}
\end{lemma}

\begin{proof}[Proof of Lemma \ref{minsker-thm}]
The result on the geometric median is the exactly the same as Corollary 4.2 of \cite{Minsker2015}.
The case for the vector of coordinate-wise medians follows from combining equation (4.4) in  \cite{Minsker2015} with Proposition 4.1 of \cite{Minsker2015}.
\end{proof}

Lemma \ref{minsker-thm} indicates that $\mathcal{S}$ in our setup can be chosen as 
\begin{align*}
(\hat{\mu} - \mu)'(\hat{\mu} - \mu) \leq r_n^2,
\end{align*}
or
\begin{align*}
(\hat{\mu}_\ast - \mu)'(\hat{\mu}_\ast - \mu) \leq r_{n,\ast}^2,
\end{align*}
either of which gives the bound with probability at least $1-2\delta$.
The former produces a tighter bound than the latter only when the dimension of $\mu$ is sufficiently high. 
Note that \eqref{Minsker-ineq} requires the existence of fourth moments due to the fact that $\text{tr}(\Sigma)$ is estimated by the median of means as well. The inequality in \eqref{Minsker-ineq} is a relatively mild condition when $n$ is  large.

\section{An Additional Empirical Example}\label{sec:Exmp}

\subsection{Bounding the Average of Log Weekly Wages}

Let $Y_i^\ast$ denote the log weekly wage and $Z_i$ the years of schooling. 
Suppose that 
\begin{align}\label{model-np} 
Y_{i}^\ast =  h( Z_{i} ) + e_{i},
\end{align}
where $h: \mathbb{R} \mapsto \mathbb{R}$ is an unknown function and  the error term $e_{i}$ satisfies $E[e_{i}|Z_i] = 0$ almost surely. 

In the 1960 and 1970 U.S. censuses, 
hours of work during a week  are measured in brackets. As a result, weekly wages are only available in terms of interval data.     
In view of this, assume that for each individual $i$, we do not observe $Y_i^\ast$, but only the interval data $[L_i, U_i]$ such that  $Y_i^\ast \in [L_i, U_i]$ along with $Z_i$. Here, $L_i$ and $U_i$ are random variables.

Assume that the support of $Z_{i}$ is finite, that is,  $\mathcal{Z} \in \{z_1, \ldots, z_J\}$.
Denote the values of $h(\cdot)$ on  $\mathcal{Z}$ by $\{ \psi_1, \ldots, \psi_J \}$. That is, $h(u) = \sum_{j=1}^J  \psi_j 1( u = z_j)$ for $u \in \mathcal{Z}$.
Suppose that the object of interest is the value of $\psi^\ast \equiv h(z^\ast)$, where  $z^\ast$ is not in the support of $Z_i$ but $z_{j-1} < z^\ast < z_j$ for some $j$.    This type of extrapolation problem is given as a motivating example in \citet[pp. 4-5]{Manski:book}.

To partially identify $\psi^\ast$, assume that $h(\cdot)$ is monotone non-decreasing. Specifically, we impose the monotonicity on $\mathcal{Z} \cup \{ z^\ast \}$.
That is, $h(z_1) \leq h(z_2)$ whenever $z_1 \leq z_2$ for any $z_1, z_2 \in  \mathcal{Z} \cup \{ z^\ast \}$.
In addition, we have the following inequality constraints:
\begin{align}\label{ineq-restriction} 
E[L_{i}|Z_{i}=z_j] \leq  \psi_j \leq E[U_{i}|Z_{i}=z_j]
\end{align}
for any $1 \leq j \leq J$. 
Note that \eqref{ineq-restriction} alone does not provide a bounded interval on $\psi^\ast$ since $z^\ast$ is not in $\mathcal{Z}$. The monotonicity assumption combined with \eqref{ineq-restriction} provides an informative bound on $\psi^\ast$.

To write the optimization problem in our canonical form,  let $\mu$ denote the population moments of the following $\bar{X}$:
\begin{align*}
n \bar{X} = 
\begin{pmatrix}
{\sum_{i=1}^{n} L_i 1(Z_{i}=z_1)} \\
 \vdots \\
  \sum_{i=1}^{n} L_i 1(Z_{i}=z_J) \\
 \sum_{i=1}^{n} U_i 1(Z_{i}=z_1) \\
 \vdots \\
 \sum_{i=1}^{n} U_i 1(Z_{i}=z_J) \\
 \sum_{i=1}^{n} 1(Z_{i}=z_1) \\
\vdots \\
\sum_{i=1}^{n} 1(Z_{i}=z_J)
\end{pmatrix}.
\end{align*}
Then, we can rewrite the constraints \eqref{ineq-restriction}  in a bilinear form: 
\begin{align}\label{ineq-restriction-bilinar} 
E[L_{i} 1(Z_{i}=z_j)] \leq  \psi_j E[ 1(Z_{i}=z_j)]  \leq E[U_{i} 1(Z_{i}=z_j)]
\end{align}
for any $1 \leq j \leq J$. 

\subsection{Empirical Results}\label{subsec:app:results}

We use a sample extract from the U.S. 2000 census to estimate $\psi^\ast$.
The sample is restricted to males who were 40 years old with educational attainment at least grade 10,
positive wages and positive hours of work.   
In the 2000 census,  variable \texttt{WKSWORK1} is  the hours of work in integer value, whereas 
\texttt{WKSWORK2} is weeks worked in intervals. 
As a parameter of interest, we focus on the average log wage for one year of college.
To illustrate our method, in our estimation procedure, we drop all observations with those with one year of college.
The resulting sample size is $n = 15,647$.

\begin{table}[htbp]
\caption{95\% Confidence Intervals}\label{tab:emp1} 
\begin{center}
\begin{tabular}{lc}
\hline\hline
 $\mathcal{S}$ & Bounds   \\ \hline
    \texttt{Box} & $[6.01, 7.86]$  \\
    \texttt{Chi-Square} & $[6.33, 6.90]$  \\  
    \texttt{Sub-Gaussian-1} & $[6.25, 7.10]$   \\
    \texttt{Sub-Gaussian-2} & $[2.62, 12.69]$   \\
    \texttt{Sub-Gaussian-3} & $[4.24, 11.77]$   \\      
    \texttt{Vector-Sub-Gaussian} & $[6.32, 6.87]$   \\   
     \texttt{Minsker} & $[2.60, 12.69]$   \\  
  \texttt{Oracle} & $[6.56, 6.61]$ \\     
    \hline
\end{tabular}
\end{center}
\end{table}%

In estimation, 
the latent 
 $Y_i^\ast$ is  $\log$(\texttt{INCWAGE}/\texttt{WKSWORK1}),
where \texttt{INCWAGE} is 
total pre-tax wage and salary income 
and
\texttt{WKSWORK1} is weeks worked in integer value.
We observe brackets $[L_i, U_i]$ of  $\log$(\texttt{INCWAGE}/\texttt{WKSWORK2}), where \texttt{WKSWORK2} is weeks worked in intervals:
[1,13], [14,26], [27,39], [40,47], [48,49], [50,52].
The brackets $[L_i, U_i]$ are random because the numerator \texttt{INCWAGE} is random.
The set $\mathcal{Z}$ includes 7 different values: 10th grade, 11th grade, 12th grade, 1 year of college,
2 years of college, 4 years of college and 5 or more years of college. 
Since we drop observations with those with 1 year of college, 
we have a missing data problem as described in the previous section.

Table \ref{tab:emp1}  shows  estimation results. The nominal level of coverage is 0.95. 
Inference methods are the same as those used in Section~\ref{sec:Exmp:new}.
The \texttt{Oracle} refers to the confidence interval of  the average log wage for one year of college
using the actual value of weekly wages for those with one year of college.
As in the empirical example in the main text, 
the \texttt{Chi-Square} and \texttt{Vector-Sub-Gaussian} intervals are tighter than the other intervals and are only slightly wider than the \texttt{Oracle} interval.


{
\section{Additional Monte Carlo Experiments}\label{sec:Exmp:appendix}

In this part of the appendix, we report the results of additional Monte Carlo experiments.
To check sensitivity to sub-Gaussian assumptions, we now generate the outcome variable $Y$ from 
\begin{align}\label{FH-iv-dgp-appendix}
Y = \mathbb{E}[ \phi(D) | Z ] + \frac{V}{52}; \ \ V \sim \frac{\chi^2(5) - 5}{\sqrt{10}}.
\end{align}
In the main text, we used the standard normal random variable for $V$; here, it is a standardized chi-square random variable, which is not sub-Gaussian but sub-exponential. 
There were 100 replications in each Monte Carlo experiment.
Figure~\ref{fig-mc:appendix} and Table~\ref{MC:tab:appendix} show that the experimental results are similar to those in the main text.
}

\begin{figure}[htbp]
\caption{Distributions of $\hat{J}_{+}$ with a chi-square disturbance term}\label{fig-mc:appendix}
\begin{center}
\includegraphics[scale=0.5, angle=270]{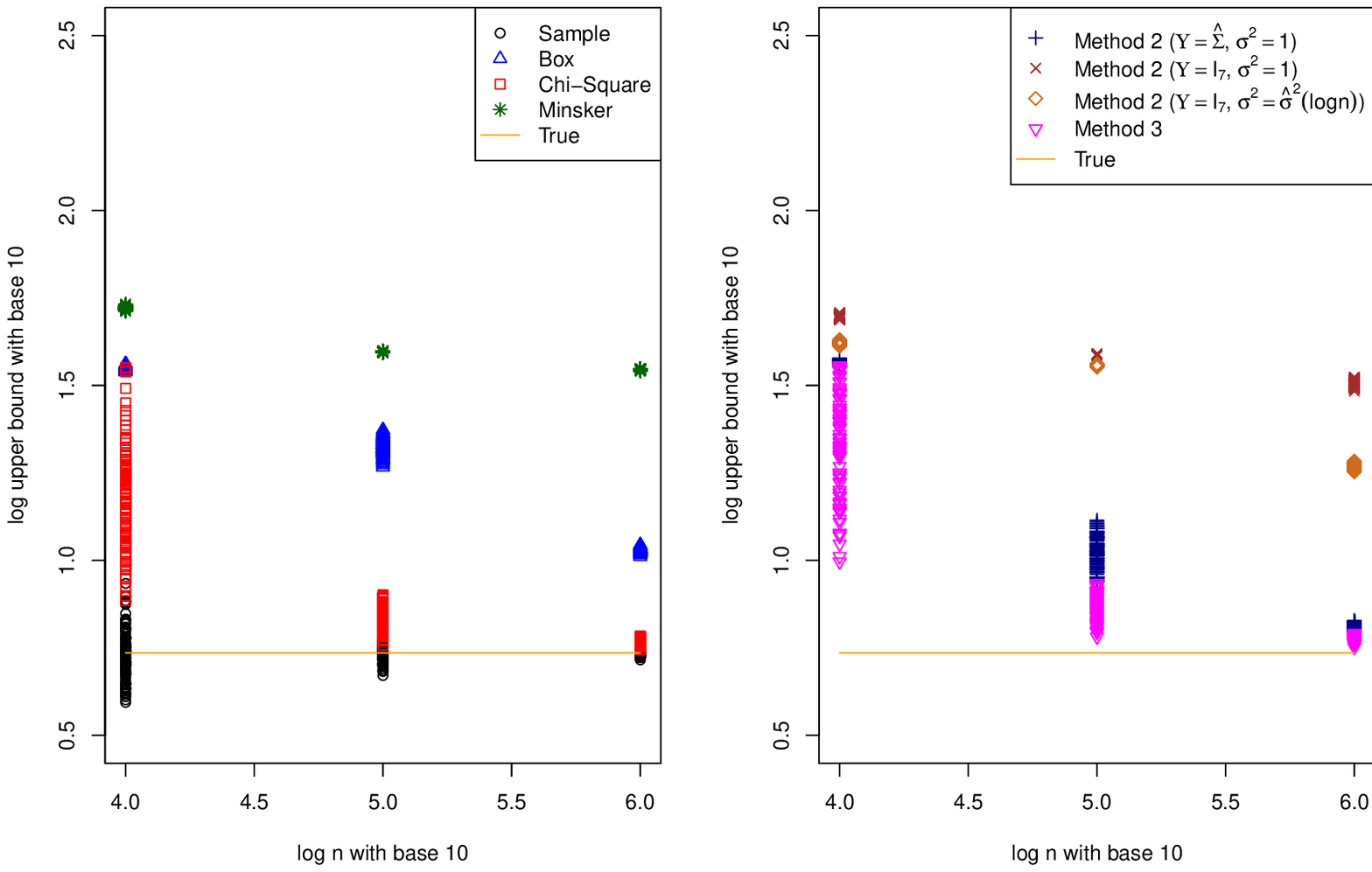}
\end{center}
\end{figure}

\begin{table}[htbp]
\caption{Results of Monte Carlo experiments with a chi-square disturbance term}\label{MC:tab:appendix}
\begin{center}

\begin{tabular}{lrrr}
 \hline\hline
 & $n = 10^4$ &  $10^5$ &  $10^6$ \\ 
  \hline
\multicolumn{4}{l}{Averages of $\hat{J}_{+}$} \\  
\texttt{Sample} & 5.45 & 5.48 & 5.44 \\ 
  \texttt{Box} & 35.56 & 21.20 & 10.69 \\ 
  \texttt{Chi-Square} & 16.87 & 6.90 & 5.83 \\ 
  \texttt{Method 2 ($\Upsilon = \hat{\Sigma}$, $\sigma^2 = 1$)} & 36.83 & 10.56 & 6.44 \\ 
  \texttt{Method 2 ($\Upsilon = I_7$, $\sigma^2 = 1$)} & 49.71 & 38.73 & 31.95 \\ 
  \texttt{Method 2 ($\Upsilon = I_7$, $\sigma^2 = \hat{\sigma}^2 (\log n)$)} & 41.74 & 36.08 & 18.49 \\ 
  \texttt{Method 3} & 24.79 & 7.37 & 5.94 \\ 
  \texttt{Minsker} & 52.59 & 39.47 & 35.15 \\ 
   \hline
\multicolumn{4}{l}{Proportions that $\hat{J}_{+} \geq J_{+} = 5.438$} \\  
\texttt{Sample} & 0.47 & 0.55 & 0.48 \\ 
  All other methods & 1.00 & 1.00 & 1.00  \\ 
 \hline   
\end{tabular}
\end{center}
\end{table}%

\newpage

{\singlespacing
\bibliographystyle{econometrica}
\bibliography{naio}
}

\end{document}